\documentclass[11pt]{article}
\pdfoutput=1
\usepackage[margin=1in]{geometry}
\usepackage{complexity}
\usepackage{mathtools}
\usepackage{xspace}
\usepackage{rpmacros}
\RequirePackage[colorlinks=true]{hyperref}
\hypersetup{
  linkcolor=[rgb]{0,0,0.4},
  citecolor=[rgb]{0, 0.4, 0},
  urlcolor=[rgb]{0.6, 0, 0}
}
\usepackage{mathpazo}
\usepackage{bbm}
\usepackage{dsfont}

\usepackage{amsthm}
\usepackage{thmtools,thm-restate}

\numberwithin{equation}{section}
\declaretheoremstyle[bodyfont=\it,qed=\qedsymbol]{noproofstyle}

\declaretheorem[name=Observation,numbered=no]{observation*}

\declaretheorem[numberlike=equation]{fact}

\declaretheorem[numberlike=equation]{theorem}

\declaretheorem[name=Theorem,numbered=no]{theorem*}

\declaretheorem[numberlike=equation]{lemma}
\declaretheorem[name=Lemma,numbered=no]{lemma*}

\declaretheorem[numberlike=equation]{corollary}
\declaretheorem[name=Corollary,numbered=no]{corollary*}

\declaretheorem[name=Proposition,numbered=no]{proposition*}

\declaretheorem[numberlike=equation]{claim}
\declaretheorem[name=Claim,numbered=no]{claim*}

\declaretheorem[name=Conjecture,numbered=no]{conjecture*}

\declaretheorem[name=Question,numbered=no]{question*}

\declaretheorem[name=Open Problem]{openproblem}

\declaretheoremstyle[bodyfont=\it,qed=$\lozenge$]{defstyle} 

\declaretheorem[numberlike=equation,style=defstyle]{definition}
\declaretheorem[unnumbered,name=Definition,style=defstyle]{definition*}

\declaretheorem[unnumbered,name=Example,style=defstyle]{example*}

\declaretheorem[unnumbered,name=Notation=defstyle]{notation*}

\declaretheorem[unnumbered,name=Construction,style=defstyle]{construction*}

\declaretheorem[numberlike=equation,style=defstyle]{remark}
\declaretheorem[unnumbered,name=Remark,style=defstyle]{remark*}

\usepackage{nth}
\usepackage{intcalc}
\usepackage{etoolbox}
\usepackage{xstring}
\hypersetup{
}

\usepackage{ifpdf}
\ifpdf
\else
\usepackage[quadpoints=false]{hypdvips}
\fi

\newcommand{\ECCCurl}[2]{http://eccc.hpi-web.de/report/\ifnumcomp{#1}{>}{93}{19}{20}#1/#2/}

\newcommand{\shortECCC}[2]{\texttt{\href{http://eccc.hpi-web.de/report/\ifnumcomp{#1}{>}{93}{19}{20}#1/#2/}{eccc:TR#1-#2}}}

\newcommand{\parseECCC}[1]{%
\StrSubstitute{#1}{TR}{}[\tmpstring]%
\IfSubStr{\tmpstring}{/}{ %
\StrBefore{\tmpstring}{/}[\ecccyear]%
\StrBehind{\tmpstring}{/}[\ecccreport]%
}{%
\StrBefore{\tmpstring}{-}[\ecccyear]%
\StrBehind{\tmpstring}{-}[\ecccreport]%
}%
\shortECCC{\ecccyear}{\ecccreport}}
 
\usepackage{algorithmicx}
\usepackage{algorithm} %
\usepackage[noend]{algpseudocode}

\algrenewcommand\algorithmicindent{1.0em}%

\newcommand{\eqdef}{\vcentcolon=}

	\renewcommand{\vec}[1]{{\mathbf{#1}}}

	\makeatletter
	\newcommand{\va}{{\vec{a}}\@ifnextchar{^}{\!\:}{}}
	\newcommand{\vb}{{\vec{b}}\@ifnextchar{^}{\!\:}{}}
	\newcommand{\vc}{{\vec{c}}\@ifnextchar{^}{\!\:}{}}
	\newcommand{\vd}{{\vec{d}}\@ifnextchar{^}{\!\:}{}}
	\newcommand{\ve}{{\vec{e}}\@ifnextchar{^}{\!\:}{}}
	\newcommand{\vy}{{\vec{y}}\@ifnextchar{^}{\!\:}{}}
	\newcommand{\vs}{{\vec{s}}\@ifnextchar{^}{\!\:}{}}
	\newcommand{\vt}{{\vec{t}}\@ifnextchar{^}{\!\:}{}}
	\newcommand{\vx}{{\vec{x}}\@ifnextchar{^}{}{}}		%
	\newcommand{\vz}{{\vec{z}}\@ifnextchar{^}{\!\:}{}}
	\newcommand{\vv}{{\vec{v}}\@ifnextchar{^}{\!\:}{}}
	\newcommand{\vu}{{\vec{u}}\@ifnextchar{^}{\!\:}{}}
	\newcommand{\vf}{{\vec{f}}\@ifnextchar{^}{\!\:}{}}
	\newcommand{\vg}{{\vec{g}}\@ifnextchar{^}{\!\:}{}}
	\newcommand{\vr}{{\vec{r}}\@ifnextchar{^}{\!\:}{}}
	\newcommand{\vw}{{\vec{w}}\@ifnextchar{^}{\!\:}{}}

	\newcommand{\vY}{{\vec{Y}}\@ifnextchar{^}{\!\:}{}}
	\newcommand{\vX}{{\vec{X}}\@ifnextchar{^}{}{}}		%
	\newcommand{\vZ}{{\vec{Z}}\@ifnextchar{^}{\!\:}{}}
	\newcommand{\vG}{{\vec{G}}\@ifnextchar{^}{\!\:}{}}
	
	\makeatother

\renewcommand{\C}{\mathbb{C}}
\renewcommand{\N}{\mathbb{N}}

\newcommand{\cC}{{\mathcal{C}}}

\newcommand{\cE}{{\mathcal{E}}}
\newcommand{\cH}{{\mathcal{H}}}

\newcommand{\cF}{{\mathcal{F}}}

\newcommand{\Nhom}{{N^{\text{hom}}_{n,r}}}

\newcommand{\codim}{\mathrm{codim}}
\newcommand{\cube}{[-1,1]_\C}
\newcommand{\volu}{\mathrm{Vol}}

\def\epsilon{\varepsilon} %
\let\eps\epsilon

\date{}

\title{A PSPACE Construction of a Hitting Set for the Closure of Small Algebraic Circuits}
\author{
Michael A. Forbes\thanks{Department of Computer Science, University of Illinois at Urbana-Champaign, Email: \texttt{miforbes@illinois.edu}.
This work was completed when the first author was at Stanford University, supported by the NSF, including NSF CCF-1617580, and the DARPA Safeware program; and when the author was at the Simons Institute for the Theory of Computing, at the University of California, Berkeley.}%
\and%
Amir Shpilka\thanks{Department of Computer Science, Tel Aviv University, Tel Aviv, Israel, E-mail: \texttt{shpilka@post.tau.ac.il}. The research leading to these results has received funding from Israel Science Foundation (grant number 552/16) and from the European Community's Seventh Framework Programme (FP7/2007-2013) under grant agreement no. 257575. Part of this work was done while the second author was at NYU.}
}
\begin{document}
\maketitle

\begin{abstract}

In this paper we study the complexity of constructing a hitting set for  $\VPbar$, the class of polynomials that can be infinitesimally approximated by polynomials that are computed by polynomial sized algebraic circuits, over the real or complex numbers. Specifically, we show that there is a $\PSPACE$ algorithm that  given $n,s,r$ in unary outputs a set of inputs from $\Q^n$ of size $\poly(n,s,r)$, with  $\poly(n,s,r)$ bit complexity, that hits all $n$-variate polynomials of degree $r$ that are the limit of size $s$ algebraic circuits. Previously it was known that a random set of this size is a hitting set, but a construction that is certified to work was only known in $\EXPSPACE$ (or $\EXPH$ assuming the generalized Riemann hypothesis). As a corollary we get that a host of other algebraic problems such as Noether Normalization Lemma, can also be solved in $\PSPACE$ deterministically, where earlier only randomized algorithms and $\EXPSPACE$ algorithms (or $\EXPH$ assuming  the generalized Riemann hypothesis) were known.

The proof relies on the new notion of a \emph{robust hitting set} which is a set of inputs such that any nonzero polynomial that can be computed by a polynomial size algebraic circuit, evaluates to a not too small value on at least one element of the set. Proving the existence of such a robust hitting set is the main technical difficulty in the proof. 

Our proof uses anti-concentration results for polynomials, basic tools from algebraic geometry and the existential theory of the reals.
\end{abstract}

\section{Introduction}\label{sec:intro}

This paper studies the following question. What is the complexity of constructing a set of points $\cH\subset \R^n$, of small bit complexity, that is guaranteed to be a hitting set for polynomials that can be infinitesimally approximated by small algebraic circuits over the real or complex numbers? 
Recall that $\cH$ is a hitting set for a class of polynomials $\cC$ if for every $f\in \cC$ there is some $\vv\in\cH$ such that $f(\vv)\neq 0$. 
The class of polynomials that can be infinitesimally approximated by poly-size algebraic circuits is commonly denoted by $\VPbar$. Thus, we ask what is the complexity of constructing a hitting set for $\VPbar$.\footnote{One can think of algebraic circuits over the reals, but our results hold for the complex numbers as well. To prove our results it will be convenient to assume that we work over algebraically closed fields, but at the end we output a set $\cH \subset \Q^n$ of small bit complexity.}

Relying on a result of Heintz and Sieveking \cite{HeintzSieveking}, who proved that the variety of efficiently computed polynomials has polynomial dimension and exponential degree, Heintz and Schnorr \cite{DBLP:conf/stoc/HeintzS80} showed that there is a poly-size hitting set for $\VPbar$, and that a random set of the appropriate polynomial size is a hitting set with high probability. If we were satisfied with a $99.9\%$ percent guarantee then picking $\cH$ at random would work  \cite{DBLP:conf/stoc/HeintzS80}. The main difficulty however is certifying that the set that we constructed is a hitting set. 

The problem of explicitly constructing an object whose existence is known by probabilistic arguments has received a lot of attention. The question is usually most difficult when there is no efficiently computable certificate to test whether a candidate construction satisfies the required properties. For example, consider $d$-regular expander graphs that have expansion larger than half the degree. By probabilistic arguments we know that random $d$-regular graphs are such expanders. Yet, we don't know of an efficiently checkable certificate that certifies such large expansion (for expansion less than $d/2$ we can use spectral methods to certify expansion).
Another such question arises in construction of Ramsey graphs. We know that when picking a graph at random from $G(n,1/2)$ (i.e. each edge is picked with probability $1/2$) with high probability it will not have cliques nor anti-cliques of size larger than, say, $3\log n$. Despite many recent advances it is not known  how to efficiently construct such graphs nor to check whether a given graph has this property. Constructing binary codes that meet the Gilbert-Varshamov bound is yet another such problem and so is the question of constructing a truth table of length $n$ that cannot be computed by boolean circuits (on $\log n$ bits) of size, say, $\sqrt{n}$. A more extreme example is that of constructing strings with large Kolmogorov complexity. A random string of length $n$ will have Kolmogorov complexity of $\Omega(n)$, but the question of deciding the Kolmogorov complexity of a string is undecidable.

We note that even when an object is known to exist via probabilistic arguments it is still not clear that it can be deterministically constructed, even in $\PSPACE$. Indeed, in $\PSPACE$ we can go over all choices of random coins for our randomized algorithm and for each construct the potential object, yet when no efficiently checkable certificate is known it is not clear how to verify that the object that we constructed have the required properties. For the questions mentioned above, of constructing expander graphs, Ramsey graphs, codes that meet the Gilbert-Varshamov bound or hard truth tables, it is clear how to check in $\PSPACE$ whether they have the required property. 

The situation is quite different when we think about hitting sets for $\VPbar$. One such difference is that for hitting sets it is impossible to go over all polynomials in $\VPbar$ as there are infinitely many such polynomials. Furthermore, we do not know an efficient way of representing them as they are limits of polynomials in $\VP$ and are not believed to have small circuits themselves: Recall that over the real or complex numbers $\VPbar$ has several equivalent definitions, the easiest may be that a polynomial $f$ is in $\VPbar$ if there exists a sequence of polynomials $\{f_i\}$ such that each $f_i$ can be computed by a size $n^c$ (for some constant $c$) algebraic circuit such that $f_i \to f$ coefficient-wise.\footnote{One can define the class $\VPbar$ over fields of positive characteristic using notion similar to border rank but we do not need this alternative definition here.} The best upper bound on the complexity of polynomials in $\VPbar$ is exponentially larger than the complexity of the approximating polynomials (which is nontrivial as the degrees could be polynomially large). See \cite{LehmkuhlL89,Burgisser04} for the exponential upper bound and also \cite{GrochowMQ16} for polynomial upper bound if one tweaks the definition of $\VPbar$ by restricting the type of allowed approximations. Thus, there is no guarantee that polynomials in $\VPbar$ have concise representation using algebraic circuits. To get a sense of why the complexity of a limit polynomial may be larger consider the tensor associated with univariate polynomial multiplication modulo $X^2$. That is, $T=x_0 y_0 z_0 + (x_1 y_0 + x_0 y_1)z_1$. It is known that the tensor rank of $T$ is $3$. However, $T$ can be represented as the limit of rank $2$ tensors: For any $\epsilon\neq 0$ consider the tensor $$T_\epsilon=\frac{1}{\epsilon}\cdot (x_0 + \epsilon x_1)\cdot (y_0+\epsilon y_1) \cdot z_1 +x_0\cdot y_0 \cdot (z_0 - \frac{1}{\epsilon}z_1) = x_0 y_0 z_0 + (x_1 y_0 + x_0 y_1)z_1 + \epsilon \cdot x_1 y_1 z_1 \;.$$ It is clear that as $\epsilon \to 0$ we get $T_\epsilon \to T$. This example shows that limits of algebraic computations can have smaller complexity than each of the polynomials in the sequence. 

The question of constructing a hitting set for $\VPbar$ that is guaranteed to work was raised by Mulmuley \cite{Mulmuley-GCT-V}. Specifically, Mulmuley asked what is the complexity of constructing a set that is guaranteed to be a hitting set for $\VP$ and $\VPbar$. For $\VP$ he observed that a hitting set can be constructed in $\PSPACE$ (when the parameters of the circuits are given in unary) and that, assuming the generalized Riemann hypothesis (GRH for short), it can be brought down to $\PH$ using a result of Koiran \cite{DBLP:journals/jc/Koiran96}. The idea is to reduce the question of checking whether a given set of points is not a hitting set to a question regarding the satisfiability of a certain set of polynomial equations in $\poly(n)$ variables and polynomial degree. I.e., to an instance of Hilbert's Nullstellensatz problem. However, for $\VPbar$ the situation is more complicated as polynomials in this class are not known (nor believed) to have small algebraic circuits. Thus, it is not clear whether the question of checking whether $\cH$ is or is not a hitting set could be reduced to Hilbert's Nullstellensatz problem. Using Gr\"obner basis, Mulmuley gave an $\EXPSPACE$ algorithm\footnote{Assuming the generalized Riemann hypothesis his algorithm can be modified to yield an $\EXPH$ algorithm using \cite{DBLP:journals/jc/Koiran96}.} for constructing such a hitting set \cite{Mulmuley-GCT-V}.  

Mulmuley raised this question due of its applicability to other questions on the borderline of geometry and complexity, mainly the so called Noether Normalization Lemma (NLL) question. He showed that constructing a normalization map could be reduced to constructing a hitting set for $\VPbar$ and thus concluded that it can be solved using randomness with a Monte Carlo algorithm, or deterministically in $\EXPSPACE$.

We note that the problem we consider is that of constructing a (qualitatively) optimal hitting set. This in particular implies a deterministic black-box PIT algorithm for $\VPbar$ which inherits the PSPACE complexity of our construction.  However, this latter result is easy to obtain directly, as one can simply evaluate the circuit over the (exponentially large) set $\{0,\ldots,r\}^n$. The correctness follows from polynomial interpolation, and one can iteratively evaluate a circuit on all points of this set in PSPACE.

\subsection{Our results}

Here we show that the problem of constructing a hitting set for $\VPbar$ is in $\PSPACE$, which bears the same consequences for the results obtained in Mulmuley's paper.

\begin{theorem}[Informal statement of main result (\autoref{thm:PIT-PSPACE})]
For integers $n,s,r$ there is an algorithm that runs in space  $\poly(n,s,r)$ and constructs a $\poly(n,s,r)$-size hitting set for all polynomials that can be infinitesimally approximated by $n$-variate homogeneous algebraic circuits of size $s$ and degree $r$.
\end{theorem}

From the work of Mulmuley it follows that $\PSPACE$ algorithms could also be devised for constructing normalizing maps (as in Noether Normalization Lemma). We refer the readers to \cite{Mulmuley-GCT-V} for more on Noether Normalization Lemma. As introducing all the relevant definitions and concepts from \cite{Mulmuley-GCT-V} requires substantial work and this is not the main focus of our work we rely on the notation of \cite{Mulmuley-GCT-V} in the statement of the next two theorems.

\begin{theorem}[NNL for {$\Delta[det,m]$} in $\PSPACE$ (see Theorem 4.1 of \cite{Mulmuley-GCT-V})]
The problem of constructing an h.s.o.p. for $\Delta[det,m]$, specified succinctly, belongs to $\PSPACE$. 
\end{theorem}

A similar result holds for other explicit varieties.

\begin{theorem}[NNL for explicit varieties in $\PSPACE$ (see Theorem 5.5 of \cite{Mulmuley-GCT-V})]
The problem of constructing an h.s.o.p. for an explicit variety $W_n$  belongs to $\PSPACE$. 
\end{theorem}

Another corollary of our result is that for every $s=\poly(n)$ we can construct in $\PSPACE$ the coefficients of an $n$-variate polynomial of constant degree that cannot even be approximated by algebraic circuits of size $s$. 

\begin{theorem}\label{thm:hard-poly}
For every constant $c$ there is a constant $c'$ such that there is a $\PSPACE$ algorithm that outputs the coefficients of an $n$ variate polynomial of degree $c'$ that is not in the closure of algebraic circuits of size $n^c$.
\end{theorem}

\begin{proof}[Sketch of proof]
The idea is to find in $\PSPACE$ a hitting set for the closure of size $s$ circuits. This hitting set has size $|\cH|=\poly(s)$. Then, by solving a system of linear equations one can find a non zero polynomial of degree roughly $\log|\cH|/\log n$, that vanishes on the points in $\cH$. By the construction of $\cH$ this polynomial is not in the closure of size $s$ algebraic circuits. 
\end{proof}

We end this part of the introduction by mentioning a natural open problem. 

\begin{openproblem}
Can the problem of constructing a hitting set for $\VPbar$ be solved in $\PH$ (assuming GRH)?
\end{openproblem}

\subsection{Sketch of proof}

The first observation is that constructing a hitting set for size $s$ and degree $r$ homogeneous circuits (i.e. for circuits in $\VP$) can be done in $\PSPACE$. The idea is that one can enumerate over all subsets of, say, $[r^2]^n$ of size, say, $(nrs)^{10}$, and for each such subset check whether there exists a circuit that computes a nonzero polynomial that vanishes over the subset. The existence of such a circuit can be checked using the universal circuit. The universal circuit $\Psi(\vx,\vy)$ is a circuit in $n$ essential variables $\vx$ and $\poly(r,s)$ auxiliary variables $\vy$ such that for any size $s$ and degree $r$ circuit $\Phi(\vx)$ there is an assignment  $\va$, to the auxiliary variables, so that the polynomials computed by $\Psi(\vx,\va)$ and $\Phi(\vx)$ are the same. Thus, if our subset is $\vv_1,\ldots,\vv_m$ we can check whether there exists a solution to $\Psi(\vv_i,\va)=0$ for all $i\in [m]$ and $\Psi(\vu,\va)=1$.\footnote{Since $\Psi(\vx,\va)$ is a homogeneous polynomial in $\vx$, if it is not identically zero then on some input $\vu$ it evaluates to $1$.} The problem of deciding whether a system of polynomial equalities has a complex solution is known as Hilbert's Nullstellensatz problem in the computer science literature. It is solvable in $\PSPACE$ and assuming the Generalized Riemann Hypothesis (GRH) it is solvable in $\PH$ (the polynomial hierarchy), see \cite{DBLP:journals/jc/Koiran96}.

We would like to use the same idea to construct hitting sets for polynomials that can be infinitesimally approximated by size $s$ and degree $r$ circuits. The problem is that even if $\cH$ is a hitting set for size $s$ and degree $r$, it may be the case that for a sequence of polynomials $\{f_i\}$, even if $f_i(\vv)\neq 0$ for all $i$, the limit polynomial may still vanish at $\vv$. Thus, it is not clear that $\cH$ also hits the closure of size $s$ and degree $r$. Indeed, consider the example given in \autoref{sec:intro}: $$T=x_0 y_0 z_0 + (x_1 y_0 + x_0 y_1)z_1$$ and $$T_\epsilon=\frac{1}{\epsilon}\cdot (x_0 + \epsilon x_1)\cdot (y_0+\epsilon y_1) \cdot z_1 +x_0\cdot y_0 \cdot (z_0 - \frac{1}{\epsilon}z_1) = x_0 y_0 z_0 + (x_1 y_0 + x_0 y_1)z_1 + \epsilon \cdot x_1 y_1 z_1 \;.$$ In addition to showing that the complexity of $T$ (measured in terms of tensor rank) is larger than that of any of the polynomials approximating it, it also demonstrates that constructing a hitting set for $\VP$ may not be sufficient for constructing a hitting set for $\VPbar$. Just to illustrate the difference consider the input $(x_0,x_1)=(y_0,y_1)=(z_0,z_1)=(0,1)$. Each of the tensors in the sequence is nonzero on this input, and indeed $T_\epsilon((0,1),(0,1),(0,1))=\epsilon\neq 0$, but in the limit we get zero, whereas the limit tensor is not the zero tensor. Thus, the input $(x_0,x_1)=(y_0,y_1)=(z_0,z_1)=(0,1)$ ``hits'' every polynomial in the sequence but the limit polynomial vanishes on it.
Thus, a hitting set for a class of polynomials $\mathcal{C}$ does not necessarily extends to the closure of $\mathcal{C}$. However, this does not rule out getting hitting sets for $\bar{\mathcal{C}}$ via hitting sets for a class of polynomials only slightly stronger than $\mathcal{C}$, or by strengthening the notion of a hitting set for $\cC$ which is what we do here.

To overcome the discrepancy between a hitting set for $\VP$ and a hitting set for $\VPbar$, we would like to find what we call a ``robust hitting set''. In a nutshell, a robust hitting set $\cH$ is such that for every polynomial $f$ that can be computed by a size $s$ and degree $r$ circuit, after an adequate normalization, there will be a point in $\cH$ on which $f$ evaluates to at least, say, $1$. Thus, if $f_i$ are all normalized and evaluate to at least $1$ on $\vv$, then if $\lim f_i = f$ then by continuity $f$ also evaluates to at least $1$ on $\vv$. Thus, $\cH$ hits $f$ as well (this idea is captured by \autoref{cla:hitting-set-continuity}).

Hence, the first step in our proof is to first prove the existence of robust hitting sets. We note that Heintz and Schnorr \cite{DBLP:conf/stoc/HeintzS80} proved the existence of a small hitting set for size $s$ and degree $r$ circuits, but their proof does not yield robust hitting sets. To prove the existence of such hitting sets we use anti-concentration results for polynomials of Carbery-Wright \cite{CarberyWright}. These results show that for a given polynomial, if we sample enough evaluation points at random, then with high probability the polynomial will evaluate to a large value on at least one of those points. This is not enough though as we cannot use the union bound since there are infinitely many circuits. What we do instead is find an $\epsilon$-net in the set of all efficiently computable polynomials. For this we use the bounds given by Heintz and Sieveking \cite{HeintzSieveking} on the dimension and degree of the algebraic variety of efficiently computable polynomials. We prove that for an algebraic variety in $\C^N$, of dimension $d$ and degree $D$, there exists an $\epsilon$-net of size roughly $D\cdot (N/\eps)^{O(d)}$. Combining the two results we are able to prove the existence of a polynomially small robust hitting set for the variety of efficiently computable polynomials. Showing the existence of robust hitting sets is the main technical difficulty in the proof.

Now that we know that robust hitting sets exist the $\PSPACE$ algorithm works as follows. It enumerates over all subsets of a relevant domain of polynomial size. For each such subset it checks whether there exists an algebraic circuit that has the right normalization (e.g. that evaluates to at least $1$ on some input from $[-1,1]^n$) and that evaluates to at most $\epsilon$ on all points in the subset. If such a solution is found then the subset is not robust and we move to the next subset. To check whether such a solution exists we need to express this system of inequalities as a formula in the language of the existential theory of the reals. Then we use the fact that formulas in this language can be decided in $\PSPACE$ to conclude that our algorithm works in $\PSPACE$.

\subsection{Organization}
The rest of the paper is organized as follows. \autoref{sec:prelim} contains some preliminaries including the notation we use throughout the paper (\autoref{sec:notation}), the definition of universal algebraic circuit (\autoref{sec:circuits}) and some basic results concerning norms of polynomials and the relation between them  (\autoref{sec:norms}). \autoref{sec:anti-con} contains results concerning anti-concentration of polynomials. In \autoref{sec:varieties} we discuss basic properties of algebraic varieties and state some results concerning the variety of polynomials computed by poly-size algebraic circuits. In \autoref{sec:eps-net} we give an upper bound on the size of $\epsilon$-net for algebraic varieties of polynomial dimension and exponential degree. Then, in \autoref{sec:hitting} we prove the existence of a robust hitting set for algebraic circuits. We discuss the existential theory of the reals in \autoref{sec:logic} and in \autoref{sec:proof} we give the $\PSPACE$ algorithm for constructing a hitting set for $\VPbar$.

\section{Preliminaries}\label{sec:prelim}

\subsection{Notation}\label{sec:notation}
We shall use the following notation. We do not mention which variety or circuit we study rather just that we shall use this parameters for every circuit or variety.
\begin{itemize}
\item $n$ is number of variables in the circuit
\item $s$ is size of circuit
\item $r$ is degree of circuit
\item $d$ is dimension of variety
\item $D$ is degree of variety
\item $\Nhom={n+r-1 \choose r}$ is number of homogeneous monomials in $n$ variables of degree $r$
\item $\vv,\vu,\ve$ points in $\R^*$
\item $\vx$ vector of variables
\item $f(\vx)$ is a polynomial and $\vf$ is its vector of coefficients
\item For $0<\delta<1$, $G_\delta = \{-1,-1+\delta,-1+2\delta,\ldots,1-2\delta,1-\delta \}^n$ is the grid
\item $\iota = \sqrt{-1}$ is the complex imaginary root of $-1$
\item  $\cube^N = [-1,1]^N + \iota \cdot [-1,1]^N = \{\va+\iota\cdot \vb \mid \va,\vb \in [-1,1]^N\}$.

\item For $0<\delta<1$, $G_\delta^\C = \{a+\iota\cdot b \mid a,b \in \{-1,-1+\delta,-1+2\delta,\ldots,1-2\delta,1-\delta \}^n\}$ is the grid in $\C$.
\item For $0<\delta<1$, $G_{\delta,r} \triangleq \{\va+k\cdot \vb \mid \va,\vb\in G_\delta \text{ and } 0\leq k\leq r\}$.
\item $\Psi,\Phi$ denote circuits
\end{itemize}

\subsection{Algebraic Circuits}\label{sec:circuits}

An \emph{algebraic circuit} is a directed acyclic graph whose leaves are labeled by either variables $x_1, \ldots, x_n$ or elements from the field $\F$,\footnote{In this paper we only consider fields of characteristic zero.} and whose internal nodes are labeled by the algebraic operations of addition ($+$) or multiplication ($\times$). Each node in the circuit computes a polynomial in the natural way, and the circuit has one or more \emph{output nodes}, which are nodes of out-degree zero. The \emph{size} of the circuit is defined to be  the number of wires, and the \emph{depth} is defined to be the length of a longest path from an input node to the output node. A circuit is called \emph{homogeneous} if every gate in it computes a homogeneous polynomial.

A useful notion is that of a universal algebraic circuit, which is a circuit that ``encodes'' all circuits of somewhat smaller size.

\begin{definition}[Universal circuit]\label{def:universal-crct}
A homogeneous algebraic circuit $\Psi$ is said to be universal for $n$-variate homogeneous circuits of size $s$ and degree $r$ if $\Psi$ has $n$ essential-inputs $\vx$ and $m$ auxiliary-inputs $\vy$, such that for every homogeneous $n$-variate polynomial $f$ of degree $r$ that is computed by an homogeneous algebraic circuit of size $s$ there exists an assignment $\va$ to the $m$ auxiliary-variables of $\Psi$ such that the polynomial computed by $\Psi(\vx,\va)$ is $f(\vx)$.
\end{definition}

The existence of efficiently computable universal circuits was shown by Raz \cite{DBLP:journals/toc/Raz10} (see also \cite{SY10}).

\begin{theorem}[Universal circuit]\label{thm:universal}
There exist constants $c_1$ and $c_2$ such that the following hold.
For any natural numbers $n,s,r$ there exists a homogeneous circuit $\Psi$ such that $\Psi$ has $n$ essential-variables, $c_1\cdot sr^4$ auxiliary-variables, degree $c_2\cdot r$ and size $c_1\cdot sr^4$,\footnote{We can assume without loss of generality that the number of auxiliary variables is the same as the size of the circuit as we can ignore some of the variables.} and it is 
universal for $n$-variate homogeneous circuits of size $s$ and degree $r$. Furthermore, for any polynomial $f(\vx)$ that can be computed by $\Psi$ and any constant $\alpha$, the polynomial $\alpha \cdot f$ can also be computed by $\Psi$.
\end{theorem}

\subsection{Norms of polynomials}\label{sec:norms}

\begin{definition}[Norm of a polynomial]\sloppy
For an $n$-variate polynomial $f(\vx)\in \R[\vx]$ we denote $$\|f\|_2\eqdef \left(\int_{[-1,1]^n} |f(\vx)|^2 d\mu(\vx)\right)^{1/2} = \left(\cE_\mu [f^2]\right)^{1/2}\,,$$ where $\mu(\vx)$ is the uniform probability measure on $[-1,1]^n$. We also denote  $$\|f\|_\infty = \max_{\vv\in [-1,1]^n}|f(\vv)|.$$
\end{definition}

\begin{remark}\label{rem:Euclidean-norm}
We shall also need to work with the usual Euclidean norm of vectors. To avoid confusion we shall denote the usual Euclidean norm of a vector $\vv$ with $\|\vv\|$. I.e., we omit the subscript when dealing with the Euclidean norm.
\end{remark}

We will need some basic results relating the $L_\infty$ norm of a polynomial to its $L_2$ norm. We start by stating a result of Wilhelmsen that generalizes a classical result by Markov for univariate polynomials.

\begin{theorem}[Multivariate Markov's theorem \cite{Wilhelmsen}]\label{thm:Markov}
Let $f:\R^n\to\R$ be a homogeneous polynomial of degree $r$, that for every $\vv \in [-1,1]^n$ satisfies $|f(\vv)|\leq 1$. Then, for every $\|\vv\| \leq 1$ it holds that $\|\nabla(f)(\vv)\|\leq 2r^2$.
\end{theorem}

Denote with $B(n,\alpha,\vu)$ the $n$-dimensional ball of radius $\alpha$ around $\vu$ and with $\volu(n,\alpha)$ its volume.

\begin{corollary}\label{cor:infty-to-one}
Let $f:\R^n\to\R$ be a homogeneous polynomial of degree $r$. Then, 
$$\|f\|_2 \geq \frac{1}{2^{2n+2}}\|f\|_\infty \cdot \volu(n,\frac{1}{4r^2}) \;.$$
\end{corollary}

\begin{proof}
By normalizing $f$ it is enough to prove the result for the case $\|f\|_\infty=1$.

Let $\vu\in  [-1,1]^n$ be such that $1=\|f\|_\infty = |f(\vu)|$. Assume further, w.l.o.g. that $f(\vu)=1$. Consider the intersection $A=B(n,\frac{1}{4r^2},\vu) \cap [-1,1]^n$. We have that $\mu(A) \geq \frac{1}{4^n}\cdot \volu(n,\frac{1}{4r^2}) $, where  $\mu(\vx)$ is the uniform measure on $[-1,1]^n$. Indeed, $\mu$ scales down by a factor of $2^n$ the usual measure of $A$, and it is immediate that $A$ contains at least a fraction $2^{-n}$ of $\volu(B(n,\frac{1}{4r^2},\vu))$.

\autoref{thm:Markov} implies that for any $\vv \in A$, it holds that $f(\vv)\geq \frac{1}{2}$. Indeed, this follows immediately from the bound on the gradient of $f$, the assumption that $f(\vu)=1$ and the fact that $\|\vu-\vv\|\leq \frac{1}{4r^2}$. Thus, we get that 
$$\|f\|_2 = \int |f(\vx)|^2 d\mu(\vx) \geq \frac{1}{4}\mu(A) \geq 
\frac{1}{2^{2n+2}}\cdot \volu(n,\frac{1}{4r^2}) \;.$$
\end{proof}

We shall also need the following lower bound on the norm of $f$ when it has at least one not too small coefficient.

\begin{lemma}\label{lem:norm-coeff}
Let $f$ be an $n$-variate homogeneous of degree $r$.
Assume that one of the coefficients in $f$ is at least $\alpha$ in absolute value. Then $\|f\|_2 \geq \alpha \cdot 2^{n/2} \cdot e^{-r}$.
\end{lemma}

The proof will use Legendre polynomials as a basis for the space of polynomials. Recall that Legendre polynomials $\{L_k(x)\}$ are univariate polynomials such that $\deg(L_k)=k$ and $$\int_{-1}^{1} L_k(x)\cdot L_m(x) d\mu(\vx) = \delta_{k=m} \cdot \frac{2}{2k+1}.$$ For an exponent vector $\bar{e}=(e_1,\ldots,e_n)$ we denote $L_{\bar{e}}(\vx) \eqdef  \prod_{i=1}^{n}L_{e_i}(x_i)$. It is again easy to see that when we run over all $\bar{e}$ we get an orthogonal family of polynomials. For a polynomial we denote with $f = \sum_{\bar{e}} c_{\bar{e}} \cdot \prod_{i=1}^{n}x_i^{e_i}$ its usual monomial expansion and with $f = \sum_{\bar{e}} \ell_{\bar{e}} \cdot L_{\bar{e}}$ its expansion with respect to the Legendre basis.
We shall also need the fact that the coefficient of $x^k$ in $L_k(x)$ is $\frac{1}{2^k} \cdot {2k \choose k}$. For more on Legendre polynomials see e.g.  \cite{Sansone}.

\begin{claim}\label{cla:legendre-coeff}
Let $f(\vx)$ be a homogeneous polynomial of degree $r$. Then for any exponent vector $\bar{e^0}=(e^0_1,\ldots,e^0_n)$ such that $\sum_i e^0_i =r$ we have that 
$$\ell_{\bar{e^0}} = c_{\bar{e^0}} \cdot  \prod_{i=1}^{n} 2^{e^0_i} \cdot \frac{1}{{2e^0_i \choose e^0_i}} .$$
In particular  $\ell_{\bar{e^0}} \geq c_{\bar{e^0}} $.
\end{claim}

\begin{proof}
From the properties above it is not hard to see that $$\ell_{\bar{e^0}} = \prod_{i=1}^{n}\frac{2e^0_i +1}{2} \int_{[-1,1)^n} f \cdot L_{\bar{e^0}} d\mu(\vx) = \prod_{i=1}^{n}\frac{2e^0_i +1}{2} \sum_{\bar{e}}  \int_{[-1,1)^n} c_{\bar{e}} \cdot \prod_{i=1}^{n}x_i^{e_i} \cdot L_{\bar{e^0}} d\mu(\vx) $$
Since $f$ is homogeneous, any exponent vector $\bar{e}\neq \bar{e^0}$ appearing in the equation has a coordinate $i$ with $e_i < e^0_i$. As $L_{e^0_i}(x_i)$ is perpendicular to all lower degree polynomials we get that 
$$= \prod_{i=1}^{n}\frac{2e^0_i +1}{2} \int_{[-1,1)^n} c_{\bar{e^0}} \cdot \prod_{i=1}^{n}x_i^{e^0_i} \cdot L_{\bar{e^0}} d\mu(\vx). $$
By the same reasoning can add lower degree terms to $\prod_{i=1}^{n}x_i^{e^0_i} $ to get the polynomial $b \cdot L_{\bar{e^0}}$ where $b$ is the inverse of the product of the leading coefficients of the $L_{e^0_i}$. That is, $b = \prod_{i=1}^{n} 2^{e^0_i} \cdot \frac{1}{{2e^0_i \choose e^0_i}}$.
Hence
$$ =  \prod_{i=1}^{n}\frac{2e^0_i +1}{2} \cdot c_{\bar{e^0}} \cdot  \prod_{i=1}^{n} 2^{e^0_i} \cdot \frac{1}{{2e^0_i \choose e^0_i}} \cdot   \int_{[-1,1)^n}  L_{\bar{e^0}}^2 d\mu(\vx) =  
c_{\bar{e^0}} \cdot  \prod_{i=1}^{n} 2^{e^0_i} \cdot \frac{1}{{2e^0_i \choose e^0_i}} .$$

\end{proof}

We now prove \autoref{lem:norm-coeff}.

\begin{proof}[Proof of \autoref{lem:norm-coeff}] 
Let  $f = \sum_{\bar{e}} \ell_{\bar{e}} \cdot L_{\bar{e}}$  be the expansion of $f$ in the Legendre basis. From orthogonality we get that
$$\cE_\mu[f^2] = \sum_{\bar{e}} \cE_\mu [ \ell_{\bar{e}}^2 \cdot L_{\bar{e}}^2] = \sum_{\bar{e}} \ell_{\bar{e}}^2 \cdot  \prod_{i=1}^{n}\frac{2}{2e_i +1} \geq \alpha^2 \cdot \frac{2^n}{(2r/n+1)^n}  \geq \alpha^2 \cdot 2^n \cdot e^{-2r}.$$
\end{proof}

Finally, we will need the following simple result connecting the usual Euclidean norm of the vector of coefficients of a polynomial and its $\| \|_2$ norm.

\begin{lemma}\label{lem:two-norms}
Let $f(\vx)$ be an $n$-variate polynomial with $S$ monomials. Let $\vf\in \R^S$ be its vector of coefficients. Then, 
$$\|f\|_2 \leq \|f\|_\infty \leq \|\vf\| \cdot \sqrt{S} .$$
\end{lemma}

\begin{proof}
Let $f = \sum_M c_M \cdot M$ be the  representation of $f$ as sum of monomials. We have that for all $\vv\in [-1,1]^n$
$$|f(\vv)| = |\sum_M c_m M(\vv)| \leq \sum_M |c_M| \leq \left(\sum_M |c_M|^2 \right)^{1/2} \cdot \sqrt{S} = \|\vf\| \cdot \sqrt{S}.$$
\end{proof}

\section{Anti-concentration results for polynomials}\label{sec:anti-con}

We will rely on the following theorem of Carbery-Wright (see Theorem 8 in \cite{CarberyWright}). 

\begin{theorem}[Carbery-Wright]\label{thm:CW-general}
There exists an absolute constant $C$ such that if $f:\R^n\to \R$ is a polynomial of degree at most $r$, $0<q<\infty$, and $\mu$ is a log-concave probability measure on $\R^n$, then, for $\alpha>0$, it holds that
$$\left(\int|f(x)|^{q/r}d\mu(\vx) \right)^{1/q} \cdot \mu\{ \vv\in\R^n \mid |f(\vv)| \leq\alpha\} \leq C\cdot q\cdot \alpha^{1/r}.$$
\end{theorem}

We give a version specialized to our purposes. 

\begin{theorem}[Carbery-Wright]\label{thm:CW}
There exists an absolute constant $C_{CW}$ such that if $f:\R^n\to \R$ is a polynomial of degree at most $r$, and $\|f\|_2=1$ then, for $\alpha>0$, it holds that 
$$\Pr_{\vv\in_U [-1,1)^n}\left[|f(\vv)|\leq\alpha\right] \leq C_{CW}\cdot r\cdot \alpha^{1/r}.$$
\end{theorem}

\begin{proof}
Apply \autoref{thm:CW-general} for  $\mu$ the uniform measure on $[-1,1)^n$ and $q=2\deg(f)$. Observe that $$\int_{[-1,1)^n} |f(x)|^{q/r} d\mu(\vx) = \int_{[-1,1)^n} |f|^2 d\mu(\vx) = \|f\|_2^2=1.$$ 
\end{proof}

We need a discrete version of this theorem which we state below. Let $\delta >0$ be such that $1/\delta$ is an integer. Recall that $G_\delta = \{-1,-1+\delta,-1+2\delta,\ldots,1-2\delta,1-\delta \}^n$. 

\begin{theorem}[Discrete Carbery-Wright]\label{thm:CW-discrete}
Let $C_{CW}$ be the constant guaranteed in \autoref{thm:CW}. Let $\delta >0$ be such that $1/\delta$ is an integer. If $f:\R^n\to \R$ is a homogeneous polynomial of degree at most $r$ with $\|f\|_2=1$ then, for $\alpha>0$, it holds that
$$\Pr_{\vv\in_U G_\delta}\left[|f(\vv)|\leq\alpha- \delta \cdot (8nr^2)^{n+1} \right]\leq C_{CW}\cdot r\cdot \alpha^{1/r}.$$
\end{theorem}

For the proof of the theorem it will be helpful to think of the uniform distribution on $G_\delta$ as generated in the following way: we first sample a point  $\vv\in [-1,1)^n$ uniformly at random and then round each coordinate $v_i$ to $m_i\delta$ for some integer $m_i$ such that $m_i\delta \leq v_i < (m_i+1)\delta$. 

To prove the theorem we need the following simple result regarding polynomials.

\begin{lemma}\label{lem:continuity}
Let $f:\R^n\to \R$ be a homogeneous polynomial of degree at most $r$. Let $\delta >0$ be such that $1/\delta$ is an integer. Let $\vv\in [-1,1)^n$ and $\vu$ be obtained from $\vv$ by the rounding process described above (i.e. rounding each coordinate $v_i$ to the largest integer multiple of $\delta$ that is smaller than or equal to $v_i$). Then $|f(\vv)- f(\vu) |\leq  
 \delta \cdot (8nr^2)^{n+1} \cdot \|f\|_2$.
\end{lemma}
\begin{proof}
By the mean value theorem there exists a point $\vw$ on the line segment connecting $\vu$ and $\vv$ such that $|f(\vv)- f(\vu) | = \|\vu-\vv\|\cdot |f'(\vw)|$, where $f'(\vw)$ is the derivative of $f$ in direction $\vu-\vv$ evaluated at $\vw$. From \autoref{thm:Markov} it follows that $|f'(\vw)|\leq2\cdot \|f\|_\infty \cdot r^2$. \autoref{cor:infty-to-one} implies that 
$$|f(\vv)- f(\vu) | = \|\vu-\vv\|\cdot |f'(\vw)| \leq  2\cdot\|\vu-\vv\|\cdot \|f\|_\infty \cdot r^2 \leq 2\cdot \|\vu-\vv\|\cdot r^2 \cdot \|f\|_2 \cdot 2^{2n+2} \cdot \frac{1}{\volu(n,\frac{1}{2r^2})}\,.$$
As $ \|\vu-\vv\| \leq \delta \cdot \sqrt{n}$ and $\volu(n,\frac{1}{2r^2}) \geq \left(\frac{1}{2nr^2}\right)^n$ we get that 
$$|f(\vv)- f(\vu) | \leq \delta \cdot \sqrt{n} \cdot r^2 \cdot \|f\|_2 \cdot 2^{2n+3} \cdot  (2nr^2)^n
\leq \delta \cdot (8nr^2)^{n+1} \cdot \|f\|_2\;.$$
\end{proof}

\begin{corollary}\label{cor:continuity}
Let $f:\R^n\to \R$ be a polynomial of degree at most $r$ with $\|f\|_2=1$. Let $\delta >0$ be such that $1/\delta$ is an integer. Assume that for some $\vv\in [-1,1]^n$, $|f(\vv)|> \alpha$ and that $\vu$ is obtained from $\vv$ by the rounding process described above. Then $|f(\vu)|> \alpha -  \delta \cdot (8nr^2)^{n+1} $. 
\end{corollary}

We now give the proof of \autoref{thm:CW-discrete}.

\begin{proof}[Proof of \autoref{thm:CW-discrete}]
We use the sampling procedure described above to sample a point $\vu\in G_\delta$. That is, we first pick $\vv\in [-1,1)^n$ at random and then round it to $\vu$. By \autoref{thm:CW} with probability at least $1-C_{CW}\cdot r\cdot \alpha^{1/r}$, $\vv$ is such that $|f(\vv)|>\alpha$. By \autoref{cor:continuity}, for any such $\vv$ we have that $|f(\vu)|>\alpha- \delta \cdot (8nr^2)^{n+1}$. Thus, the probability that we sample $\vu$ with $|f(\vu)|\leq \alpha-  \delta \cdot (8nr^2)^{n+1} $ is at most 
$C_{CW}\cdot r\cdot \alpha^{1/r}$.
\end{proof}

We will be working over the complex numbers and so we need to slightly adjust the results above. Let $f:\C^n\to \C$ be a homogeneous polynomial of degree $r$. 
Consider the real and imaginary parts of $f$, $\Re(f)$ and $\Im(f)$, respectively. We can view both as polynomials $\Re(f),\Im(f):\R^{2n}\to \R$. That is, for every $\va,\vb \in \R^n$, $f(\va+\iota \vb) = \Re(f)(\va,\vb)+\iota \cdot \Im(f)(\va,\vb)$. It is clear that both $\Re(f)$ and $\Im(f)$ are homogeneous polynomials of degree $r$ as well. We define $\|f\|_2 \triangleq \|\Re(f)\|_2 + \|\Im(f)\|_2$.
As we work over the complex numbers we will refer to the set $G_\delta^\C = \{\va+\iota\cdot \vb \mid \va,\vb \in G_\delta\}$.

\begin{theorem}[Discrete Carbery-Wright over $\C$]\label{thm:CW-discrete-complex}
Let $C_{CW}$ be the constant guaranteed in \autoref{thm:CW}. If $f:\C^n\to \C$ is a homogeneous polynomial of degree at most $r$ with $\|f\|_2=1$ then, for $\alpha>0$, it holds that $$\Pr_{\vv\in_U G_\delta^\C}\left[|(f)(\vv)|\leq \alpha- \frac{1}{2}\delta \cdot (16nr^2)^{2n+1}  \right]\leq C_{CW}\cdot r\cdot (2\alpha)^{1/r}.$$
\end{theorem}

\begin{proof}
As $\|f\|_2=1$ either $ \|\Re(f)\|_2\geq 1/2$ or $ \|\Im(f)\|_2\geq 1/2$. Assume without loss of generality the former happens. Note also that if $|\Re(f)(\va,\vb)| > \gamma$ then $|f(\va+\iota \cdot \vb)| > \gamma$. \autoref{thm:CW-discrete} implies that the probability that we sample $\va,\vb\in G_\delta$ with $|\Re(f)(\va,\vb)|\leq 2\alpha-  \delta \cdot (16nr^2)^{2n+1} \cdot \|\Re(f)\|_2$ is at most 
$C_{CW}\cdot r\cdot (2\alpha)^{1/r}$. Thus, with probability at least $1-C_{CW}\cdot r\cdot (2\alpha)^{1/r}$ we get that $(\va+\iota \vb)$ is such that 
$$|f(\va+\iota \cdot \vb)| > 2\alpha-  \delta \cdot (16nr^2)^{2n+1} \cdot \|\Re(f)\|_2 \geq \frac{1}{2}(2\alpha-  \delta \cdot (16nr^2)^{2n+1}).$$

\end{proof}

\section{The algebraic variety of small algebraic circuits}\label{sec:varieties}

We start by providing some basic definitions from algebraic geometry. For more on algebraic geometry see \cite{CLO}. We follow essentially the same treatment given in \cite{DBLP:conf/stoc/HeintzS80,HeintzSieveking}.

\begin{definition}[Basic AG definitions]
A subset $V\subseteq \C^n$ is called (Zariski-)closed\footnote{Over $\C$ if a set is closed in the Zariski topology then it is also closed in the usual Euclidean topology.} if there exists a set of polynomials $\cF\subseteq \C[\vx]$ such that $V = \{ \vv \in \C^n \mid \forall f\in \cF, \; f(\vv)=0\}$. The closed sets define the Zariski topology of $\C^n$. The closure of a set $V\subseteq \C^n$ is the intersection of all closed sets containing $V$. A closed set $V$ is called irreducible if for any two closed sets $V_1,V_2$ such that $V_1 \cup V_2 = V$ it holds that either $V_1=V$ or $V_2=V$.
\end{definition}

Closed sets are also called varieties. Irreducible closed sets are irreducible varieties. 

\begin{definition}[Dimension]
The dimension of an irreducible variety $V$, denoted $\dim(V)$ is the maximal integer $m$ such that there exist $m$ irreducible varieties $\{V_i\}$ satisfying $\emptyset \subsetneq V_1\subsetneq V_2 \subsetneq \ldots \subsetneq V_m \subsetneq V$. The dimension of a reducible variety is the maximal dimension of its irreducible components.
\end{definition}

We now give some basic facts.
\begin{fact}
\begin{enumerate}
\item Each variety $V$ can be represented uniquely as a minimal finite union of irreducible varieties. Each irreducible set in this representation is called a component of $V$. 
\item The dimension of every variety $\emptyset \neq V$ is finite.

\item If $U,V$ are varieties, where $U$ is irreducible and $U\not\subseteq V$ then $\dim(U\cap V)<\dim(U)$.
\end{enumerate}
\end{fact}

It is a basic fact that any irreducible variety over $\C$ is connected (as a complex manifold). 

\begin{theorem}[Irreducible varieties are connected]\label{thm:var-connected}
Every irreducible variety $V\subseteq \C^N$ is connected as a topological space (in the usual Euclidean topology). 
\end{theorem}

\begin{proof}
The claim follows immediately from Theorem 1 in Chapter VII section 2.2 of  \cite{Shafarevich-II}, noting that in our variety every point is closed (thus, $X(\C)$ in the statement of Theorem 1 there is our irreducible variety).
\end{proof}

Another important definition is that of a \emph{degree} of a variety.

\begin{definition}[Degree]
The degree of an irreducible variety $V\subseteq \C^n$, denoted $\deg(V)$, is the maximal cardinality of a finite intersection of $V$ with an affine linear space. That is,
$$\deg(V) = \max\left\{ \left| V \cap A \right| \mid A\subset \C^n \; \text{is an affine linear space, and } \left| V \cap A \right|<\infty \right\}\;.$$
When $V$ is not irreducible, let $V=\cup_i V_i$, where $V_i$ are the irreducible components of $V$. We define $\deg(V)$ as $$\deg(V) = \sum_i \deg(V_i) \;.$$
\end{definition}

We will rely on the following estimates of Heintz and Sieveking. 

\begin{theorem}[Variety of easy polynomials \cite{HeintzSieveking}]\label{thm:HS-bound}
For every natural numbers $n,s,r$ there exists a set $W(n,s,r)\subseteq \C^\Nhom$ such that $W(n,s,r)$ contains the coefficient vectors of all $n$-variate homogeneous polynomials, $f\in \C[\vx]$, of degree $r$, that can be computed by homogeneous algebraic circuits of size at most $s$. Furthermore, 
$$\dim(W(n,s,r)) \leq (s+1+n)^2$$
and 
$$\deg(W(n,s,r))\leq (2sr)^{(s+1+n)^2}\;.$$
\end{theorem}

\begin{remark}
We note that the result above holds not only for homogeneous polynomials and can be slightly improved if we restrict our attention to the homogeneous case as we do here, but this is not crucial for our purposes.
\end{remark}

\begin{remark}
We note that the main message behind \autoref{thm:HS-bound} is that the dimension of the ambient space, $\Nhom$, does not appear in the upper bounds on $\dim(W(n,s,r))$ and $\deg(W(n,s,r))$.
\end{remark}

To prove our main result it will be convenient to consider the universal circuit. As the universal circuit for $n$-variate homogeneous circuits of size $s$ and degree $r$ has size $O(sr^4)$ we obtain the following immediate corollary. Note that when speaking of the polynomials that can be computed by the universal circuit we think of the set of polynomials that is obtained by running over all assignments to the auxiliary variables. Indeed, for any such assignment the circuit that is obtained is homogeneous in its essential variables and of size $O(sr^4)$.

\begin{theorem}[Variety of projection of the universal circuit]\label{thm:variety-universal}
For all natural numbers $n,r,s$ there exists a set $V(n,s,r)\subseteq \C^\Nhom$ such that $V(n,s,r)$ contains the coefficient vectors of all homogeneous polynomials of degree $r$ that can be computed by the universal circuit for $n$-variate homogeneous circuits of size $s$ and degree $r$. Furthermore, there exists a constant $c$ such that
$$\dim(V(n,s,r)) \leq c\cdot (sr^4+1+n)^2$$
and 
$$\deg(V(n,s,r))\leq (c sr^5)^{c\cdot (sr^4+1+n)^2}\;.$$
\end{theorem}

To ease notations we shall use the following corollary.

\begin{corollary}\label{cor:variety-universal}
Let $V(n,s,r)$ be as in \autoref{thm:variety-universal}. Then, there exists a constant $c_{\ref{cor:variety-universal}}$ such that 
$$\dim(V(n,s,r)) \leq(srn)^{c_{\ref{cor:variety-universal}}}$$
and 
$$\deg(V(n,s,r))\leq 2^{(srn)^{c_{\ref{cor:variety-universal}}}}\;.$$
\end{corollary}

\autoref{thm:variety-universal} speaks about a variety containing coefficient vectors of easy polynomials. As varieties are closed, the same variety also contains all coefficient vectors of polynomials that are limits of easy polynomials. 

\begin{definition}[Closure of easy polynomials]
A homogeneous polynomial $f\in \C[\vx]$ is in the closure of size $s$ and degree $r$ algebraic circuits if there exists a sequence of $n$-variate, degree $r$, homogeneous polynomials $\{f_i(\vx)\}$, such that each $f_i$ can be computed by a homogeneous circuit of size $s$ and degree $r$, and $\lim_{i\to \infty} f_i = f$. In other words, there exists a sequence of homogeneous algebraic circuits of degree $r$ and size $s$ such that the coefficients vector of the polynomials they compute converge to the coefficient vector of $f$.
\end{definition}

\begin{remark}
At first sight it may seem, thanks to the definition of the universal circuit $\Psi(\vx,\vy)$, that if $f \in  \VPbar$ that $f$ also has a small circuit. Indeed, if $f_i \to f$ and $f_i = \Psi(\vx,\va_i)$ then it seems that $f = \Psi(\vx,\lim_{i\to \infty} \va_i)$. The problem with this argument however is that $\{\va_i\}$ may not converge. An example for this phenomenon may be seen in the the difference between the border rank of a tensor and its rank. It may be the case that a tensor has border rank $r$ yet it's rank may be larger. See e.g. \autoref{sec:intro} in this paper and Section 6 in \cite{DBLP:journals/toc/Blaser13}.
\end{remark}

\begin{corollary}
The variety $V(n,s,r)$ defined in \autoref{thm:variety-universal} contains all coefficient vectors of homogeneous polynomials that are  in the closure of size $s$ and degree $r$ algebraic circuits.
\end{corollary}

Finally, we define a notion that will be useful in the upcoming proofs.

\begin{definition}[Axis-parallel random variety]
We say that a variety $V$ is axis-parallel random if for any axis-parallel affine subspace $A$ (i.e. a subspace defined by setting some coordinates to constants) it holds that $\dim(V\cap A) \leq \dim(V) - \codim(A)$.
\end{definition}

One way to think of this definition is that a variety is axis-parallel-random if by restricting a variable to a constant we move to a strictly smaller subvariety. 

It is clear that if $V$ is axis-parallel random then for every axis-parallel affine subspace $A$, $V\cap A$ is also axis-parallel random. Next we show that by slightly perturbing a variety makes it an axis-parallel random one. That is, we will show that for a linear transformation $T$, the variety $T(V)$ is axis-parallel random.

\begin{theorem}[A random perturbation makes a variety axis-parallel random]]\label{thm:T-axis-random}
Let $0<\delta$ and let $T=I_N+A$, where $I_N$ is the $N\times N$ identity matrix and $A$ is a random matrix where each $a_{i,j}$ is chosen independently uniformly at random from $[0,\delta]$. 
Let $V\subseteq \C^N$ be a variety of dimension $d$.
Then $T(V)$ with probability $1$ $T(V)$ is axis-parallel random.
\end{theorem}

To prove the theorem we will need the following theorem that characterizes the dimension of a variety in terms of the algebraic rank of the polynomials defining that variety. See e.g. Theorem 2 in Chapter 9, $\mathsection$5 of \cite{CLO}. 
\begin{theorem}[Characterization of dimension via algebraic dependence]
Let $V\subseteq \C^N$ be a variety and let $I=I(V)$ be the ideal of all polynomials vanishing on $V$. 
Let the coordinate ring of $V$ be $\C[V] \triangleq \C[\vx]/I$. Then the dimension of $V$ equals the maximal number of elements of $\C[V]$ that are algebraically independent. 
\end{theorem}

\begin{proof}[Proof of \autoref{thm:T-axis-random}]

To prove the theorem we first show that with high probability no variable (or actually no linear function of the form $x_i=c$) will be in the ideal $I(V)$. Then we prove that restricting any variable to a constant reduces the algebraic rank of $\C[V]$. 

\begin{lemma}\label{lem:T-for-I}
 Let $L \subset I(T(V))$ be the linear space of all linear functions in $I(T(V))$. 
If $\dim(V)=d$ then for any $d$ variables $x_{i_1},\ldots,x_{i_{d}}$ the probability that there exists a non zero linear combination $\sum_{j=1}^{d}\alpha_j \cdot x_{i_j} + \alpha_0 \in I(T(V))$ is $0$.\footnote{That is, this can fail only for a set of matrices of measure zero}
\end{lemma}

\begin{proof}
First observe that $d\leq N - \dim(L)$, or $\dim(L)\leq N-d$. 
Further, observe that the effect of applying $T$ to $V$ on $L$ is essentially applying $T^{-1}$ to the linear functions in $L$. Thus, the event that we are considering checks whether a random subspace of dimension at most $N-d$ intersects a given $d$ dimensional affine space. This probability is $0$ over the reals. This holds even for a $T$ chosen as in the theorem. For example, this can be seen by noting that the two subspaces intersect iff a certain determinant is zero: the rows of the determinant will compose of the basis for the two subspaces. It is easy to see that the determinant is a nonzero polynomial in the $A$ variables and thus it is nonzero with probability $1$.  
\end{proof}

Next we wish to show that setting any $k\leq \dim(V)$ variables to constant reduces the algebraic rank of $\C[V]$ by $k$ with probability $1$. We prove this property for irreducible varieties and then conclude it to arbitrary varieties. A property that will be useful is that if $V$ is an irreducible variety then $I(V)$ is a prime ideal (see e.g., \cite{CLO}).

\begin{lemma}\label{lem:T-coor-ring}
Let $f_1,\ldots,f_d$ be algebraically independent polynomials in $\C[T(V)]$. Let $x_{i_1},\ldots,x_{i_k}$ be any $k\leq d$ different variables. Then for any $k$ field elements $\alpha_1,\ldots,\alpha_k$ restricting $x_i=\alpha_i$ reduces the algebraic rank of $\{f_i\}$ by $k$ with probability $1$ (over the choice of $T$). 
\end{lemma}

\begin{proof}
We will prove the claim by induction on $k$. For $k=1$ note that since the $f_i$ are maximally algebraically independent in $\C[T(V)]$ then for any other nonzero polynomial $g$ in $C[T(V)]$ there exists a nonzero polynomial $F_g(z_1,\ldots,z_{d+1})$ such that $F_g(f_1,\ldots,f_d,g)\equiv 0$, where by that we mean that $F_g(f_1,\ldots,f_d,g)$ is the zero element in $\C[T(V)]$, that is, $F_g(f_1,\ldots,f_d,g)\in I(T(V))$. Observe further that since $I(T(V))$ is a prime ideal (applying an invertible linear transformation does not affects the irreducibility of the variety) if $g\cdot h\in I(T(V))$ then $h\in I(T(V))$. Hence, we can assume without loss of generality that $z_{d+1}$ does not divide $F_g$. 

\sloppy From \autoref{lem:T-for-I} we know that for any $\alpha$, $g=x_{i_1}-\alpha$ is not the zero polynomial in $\C[T(V)]$. Thus, there is such polynomial $F_{g}$. As $z_{d_1}$ does not divide $F_g$ we can express it as $F_g = z_{d+1}\cdot F_1(z_1,\ldots,z_{d+1}) + F_0(z_1,\ldots,z_d)$, where $F_0\neq 0$. Thus the polynomial $g \cdot F_1(f_1,\ldots,f_d,g) + F_0(f_1,\ldots,f_d)$ is in $I(T(V))$. Now, adding the linear polynomial $g$ to $I(T(V))$ to get $I^{(1)} = ( I(T(V)),g)$ (the ideal generated by $I(T(V))$ and $g$), it follows that 
$F_0(f_1,\ldots,f_d)\in I^{(1)}$. Thus, the $\{f_i\}$ become algebraic dependent when setting $g=0$, i.e. when restricting $x_{i_1}=\alpha$.

To prove the case of general $k$ we just notice that the same argument will work thanks to \autoref{lem:T-for-I}, where at the $k$th step we consider any algebraically independent set $f'_1,\ldots,f'_{d-k+1}$ in $\C[I^{(k-1)}]$.
\end{proof}

The proof of  \autoref{thm:T-axis-random} now follows since the probability of a bad event is $0$ and each variety it the union of a finite number of irreducible components.

\end{proof}

\subsection{$\epsilon$-nets for algebraic varieties}\label{sec:eps-net}

In this section we construct $\epsilon$-nets for varieties. We shall use the Euclidean norm for this and to avoid confusion we shall denote the Euclidean norm of a vector $a$ with  it with $\| a \|$.  That is, there is no subscript $2$ when using the Euclidean norm.

\begin{definition}[$\epsilon$-net]
Let $V \subseteq \C^N$. A set $E\subseteq V$ is an $\epsilon$-net for $V$ if for every $\vv\in V$ there exists  $\ve\in E$ such that $\|\ve-\vv\| \leq \epsilon$.
\end{definition}

Recall the notation $\cube^N = [-1,1]^N + \iota \cdot [-1,1]^N = \{\va+\iota\cdot \vb \mid \va,\vb \in [-1,1]^N\}$.

\begin{theorem}[$\epsilon$-net for varieties]\label{thm:eps-net}
Let $N,d,D$ be integers such that $d<\sqrt{N}$, and $\epsilon>0$. 
Let  $V\subseteq \C^N$ be a $d$-dimensional variety of degree $D$ which is axis-parallel random. Denote $\hat{V}= V \cap \cube^N$. There exists an $\epsilon$-net $E \subseteq \hat{V}$ of size smaller or equal to $D\cdot (75N^2 /\epsilon^2)^{d+1}$. %
\end{theorem}

\begin{proof}
The proof is by induction on the dimension of $V$. If $\dim(V)=0$ then $|V|= D$ and the claim is trivial. Furthermore, it is a $\delta$-net for any $\delta>0$. Assume $d>0$. 
For $\alpha\in\C$ and $i\in [N]$ let $$H_i(\alpha) = \{\vv\in\C^N \mid \vv_i=\alpha\}.$$ That is, $H_i(\alpha)$ is the hyperplane obtained by fixing the $i$'th coordinate to $\alpha$. Let $V_i(\alpha) = V\cap H_i(\alpha)$ and $\hat{V}_i(\alpha) = V_i(\alpha)\cap \cube^N$. As $V$ is axis-parallel random, $V_i(\alpha)$ is a variety of dimension $d-1$ and degree at most $D$, which is also axis-parallel random. Let $\eta ,\delta >0$ be constants to be determined later, such that $1/\eta$ is an integer. By the induction hypothesis, there is a subset $E_i(\alpha)\subseteq \hat{V}_i(\alpha)$ which is a $\delta$-net for $V_i(\alpha)\cap \cube^N$ of size $D\cdot (75N^2 /\delta)^{(d-1)+1}$. 
Let $$E' = \bigcup_{i\in [N],\alpha,\beta \in \{-1,-1+\eta,\ldots,1-\eta,1\}}E_i(\alpha+\iota\cdot \beta).$$
It is clear that $$|E'|\leq N\cdot (2/\eta+1)^2\cdot D\cdot (75N^2 /\delta^2)^{d}.$$
The set $E'$ is almost our $\epsilon$-net. All that is left to do is to cover points of $V$ that are not close to any intersection point of $V$ with any of the hyperplanes we considered. 

Let $$H = \bigcup_{i\in [N],\alpha,\beta \in \{-1,-1+\eta,\ldots,1-\eta,1\}}H_i(\alpha+\iota\cdot \beta).$$ Consider the set $\cube^N \setminus H$. It is a union of disjoint ``cells'' whose ``walls'' have been removed (the walls being the hyperplanes). 
Recall that by \autoref{thm:var-connected} we have that an irreducible variety is connected. Thus, when considering the irreducible components of $V$, we see that each component either intersects a ``wall'' of a cell or is completely contained in it (or completely disjoint from it). 
Thus, as $V$ is of degree $D$, it has at most $D$ irreducible components and in particular, at most $D$ of $V$'s irreducible components are contained in cells. 
From each connected component that is contained in such a cell pick any point. Let $B$ be the set of points thus chosen. It follows that, $|B|\leq D$. Finally, let $E=E'\cup B$. 
We claim that for $$\eta = \frac{1}{\lceil \frac{2N}{\epsilon}\rceil} \quad \text{and} \quad \delta = \left( 1 - \frac{1}{\sqrt{2N}}\right)\cdot\epsilon \;,$$
the set $E$ is an $\epsilon$-net for $\hat{V}$. Indeed, let $\vv\in \hat{V}$ be arbitrary. Consider the connected component to which $\vv$ belongs. If this component is contained in one of the cells, then there is some $\ve\in B\subseteq E$ in the same cell as $\vv$. As each cell is contained in a cube whose diameter is $\sqrt{2N}\eta$, it holds that $\|\vv-\ve\|\leq \sqrt{2N}\eta\leq \epsilon $. On the other hand, if the connected component containing $\vv$ is not contained in any of the cells, then it must intersect the closure of the cell containing $\vv$. Let $\vu$ be this intersection point. By the construction of $E'$, there is some point $\ve$ such that $\|\ve-\vu\| \leq \delta$. As $\|\vv-\vu\|  \leq \sqrt{2N}\eta$ we get by the triangle inequality that $\|\ve-\vv\| \leq \sqrt{2N}\eta + \delta \leq \epsilon.$
To conclude the proof we note that 
\begin{eqnarray*}
|E| &\leq &  N\cdot (2/\eta+1)^2\cdot D\cdot (75 N^2 /\delta^2)^{d} + D \\ 
&\leq & N\cdot (5N/\epsilon)^2\cdot D\cdot \left(75 \left(1+\frac{1}{\sqrt{N}}\right)N^2 /\epsilon^2\right)^{d} + D \\
&\leq & (70 N^2/\epsilon^2) \cdot D\cdot \left(75 N^2 /\epsilon^2 \right)^d +D\\
&\leq & D\cdot (75  N^2/\epsilon^2) \cdot \left(75 N^2 /\epsilon^2\right)^d \\
&=&D\cdot \left(75 N^2 /\epsilon^2\right)^{d+1} ,
\end{eqnarray*}
where we have used the fact that for $N\geq 2$ and $d<\sqrt{N}$ it holds that $(1+1/\sqrt{N})^d<2.8$.
\end{proof}

We would like to apply the result of \autoref{thm:eps-net} for  $V(n,s,r)$. A small technical issue is that $V(n,s,r)$ is not axis-parallel random. Nevertheless, \autoref{thm:T-axis-random} guarantees that for a random transformation $T$ the set $T(\cH)$ is  an $(\eta/2)$-robust hitting set for $V(n,s,r)$. 

Let $T$ be as guaranteed by \autoref{thm:T-axis-random} for $V=V(n,s,r)$. Note that if $V=V(\cF)$ and we define $V'=V(\cF\circ T)$ then $V = T(V')$. Indeed, $v\in V'$ iff $(f\circ T)(\vv)=0$ for all $f\in \cF$ which is equivalent to $T\vv\in V$ or $\vv\in T^{-1}V$. 

\begin{lemma}\label{lem:T-eps}
Let $V'=T^{-1}(V)\subseteq \C^N$ be an axis-parallel random variety where $T=I+A$ such that each entry of $A$ lies in $[0,1/N^{2d}]$.\footnote{Observe that the conclusion of \autoref{thm:T-axis-random} holds for any choice of $0<\delta$, in particular for $\delta = N^{-2d}$.} Let $E'\subseteq V'$ be an $\epsilon$-net for $V'\cap \cube^N$. Then $E\eqdef T(E')\subset V$ is an $(1+1/N)\cdot \epsilon$-net for $V\cap [-(1-\frac{1}{N}),1-\frac{1}{N}]^N_\C$.
\end{lemma}

\begin{proof}
Note that by construction of $T$ we have that $T^{-1} = I +  \sum_{k=1}^{\infty} (-A)^k$ and that each entry of $T^{-1}-I$ is bounded in absolute value by, say, $\frac{1}{N^d}$. In particular, $$ T^{-1}\left([-(1-\frac{1}{N}),1-\frac{1}{N}]^N_\C\right) \subseteq \cube^N \,,$$ where, as before, we have $[-(1-\frac{1}{N}),1-\frac{1}{N}]^N_\C = [-(1-\frac{1}{N}),1-\frac{1}{N}]^N+\iota [-(1-\frac{1}{N}),1-\frac{1}{N}]^N$.\footnote{Notice that as $T$ is real it operates on the real part and the imaginary part of each vector separately.}

Let $\vv \in V\cap [-(1-\frac{1}{N}),1-\frac{1}{N}]^N_\C$. We have that $$T^{-1}(\vv) \in T^{-1}(V) \cap T^{-1}\left([-(1-\frac{1}{N}),1-\frac{1}{N}]^N_\C \right) \subseteq T^{-1}(V) \cap \cube^N = V'  \cap \cube^N  \,.$$
Let $\ve\in E'$ be such that $\|\ve'-T^{-1}(\vv)\|\leq \epsilon$. We thus have that $\|T(\ve)-\vv\| \leq \|T\| \cdot \|\ve-T^{-1}\vv\| \leq^{(*)} (1+1/N)\epsilon$, where inequality $(*)$ follows easily from the construction of $T$ (and $\|T\|$ is the operator norm of $T$). 
\end{proof}

As corollary we get the same parameters for varieties that are not necessarily axis-parallel random.

\begin{corollary}[$\epsilon$-net for varieties]\label{cor:eps-net}
Let $N,d,D$ be integers such that $d<\sqrt{N}$, and $\epsilon>0$. 
Let  $V\subseteq \C^N$ be a $d$-dimensional variety of degree $D$. Denote $\hat{V}= V \cap [-(1-\frac{1}{N}),1-\frac{1}{N}]^N_\C$. There exists an $\epsilon$-net $E \subseteq \hat{V}$ of size smaller or equal to $D\cdot (76 N^2 /\epsilon^2)^{d+1}$. %
\end{corollary}

\begin{proof}
This immediately follows from \autoref{thm:eps-net} and \autoref{lem:T-eps} by plugging $\epsilon' = \epsilon/(1+\frac{1}{N})$ to the bounds there.
\end{proof}

\section{Robust hitting sets}\label{sec:hitting}

In this section we define the notion of a robust hitting set and prove its existence.

\begin{definition}[$\epsilon$-Robust hitting set]
A subset $\cH \subseteq \C^n$ is an $\epsilon$-robust hitting set for 
a set of polynomials $V\subseteq \F[\vx]$ if for every $f\in V$
there is some $\vv \in \cH$ such that $|f(\vv)|\geq \epsilon\cdot \|f\|_2$.

Let $n,r,s$ be integers.  We say that $\cH$ is an $\epsilon$-robust hitting set for size $s$ and degree $r$ if it is an $\epsilon$-robust hitting set for the set of $n$-variate polynomials that can be computed by size $s$ and degree $r$  homogeneous algebraic circuits.
\end{definition}

The next claim shows that a robust hitting set for size $s$ algebraic circuits is also a robust hitting set for the closure of such circuits, that is for $V(n,s,r)$.

\begin{claim}\label{cla:hitting-set-continuity}
If a finite $\cH \subseteq \R^n$ is an $\epsilon(n)$-robust hitting set for size $s$ and degree $r$ then for any $f\in {V}(n,s,r)$ there is $\vv\in \cH$ such that $|f(\vv)|\geq \epsilon(n)\cdot \|f\|_2$. 
\end{claim}

\begin{proof}
Let $f\in {V}(n,s,r)$. For $i\in \N$ let $f_i \in {V}(n,s,r)$ be such that $\lim_{i\rightarrow \infty}f_i=f$ and each $f_i$ is computed by a homogeneous circuit of size $s$ and degree $r$. Clearly, it also holds that $\lim_{i\rightarrow \infty}\|f_i\|_2=\|f\|_2$. As $\cH$ is finite there is $\vv\in\cH$ at which  infinitely many $f_i$ evaluate (in absolute value) to at least $\epsilon(n)\cdot \|f_i\|_2$. Thus, there is a subsequence $f_{i_j}$ satisfying  $\lim_{j\rightarrow \infty}f_{i_j}=f$ and $|f_{i_j}(\vv)|\geq \epsilon(n)\cdot \|f_{i_j}\|_2$. By continuity it holds that $|f(\vv)|\geq \epsilon(n)\cdot \|f\|_2$ as well.
\end{proof}

In the next lemma we think of each point $\vf\in \C^\Nhom$ as being the coefficient vector of some homogeneous $n$ variate polynomial $f$ of degree $r$. That is $f(\vx) = \sum_{\deg(M)\leq r} f_M \cdot M(\vx)$, and we index coordinates of $\C^\Nhom$ with degree $r$ monomials. The lemma shows that in order to construct a robust hitting set for a set of polynomials it is enough to construct a good enough robust hitting set for an $\epsilon$-net in the variety.

\begin{lemma}\label{lem:robust-net}
Let $V\subseteq \cube^\Nhom$ be such that if $\vf\in V$ and $\alpha\vf \in \cube^\Nhom$ then $\alpha\vf\in V$.\footnote{Notice that this property holds for $V(n,s,r)\cap   \cube^\Nhom$.} 
Let $E\subseteq V$ an $\epsilon$-net for $V$. Assume that $\cH\subseteq \cube^n$ is such that for every $\vg\in E$ there exists $\vv\in \cH$ such that $g(\vv)\geq \eta \cdot \|g\|_2$, for some $\eta < 1$.
Assume further that $\eta,\epsilon,\Nhom$ and $r$ satisfy that $$10 \cdot \epsilon \cdot \sqrt{\Nhom} < \frac{1}{8}\eta \cdot 2^{n/2} \cdot e^{-r} < \frac{1}{4}.$$
Then, for every $\vf\in V$ there exists $\vv\in \cH$ such that $|f(\vv)| \geq \frac{1}{4} \eta \|f\|_2$.
\end{lemma}

\begin{proof}
Let $\vf\in V$. Assume w.l.o.g. that the maximal coefficient of $f$ is $1/2$. This can be easily obtained by multiplying $\vf$ by a field element.
Let $\vg\in E$ be such that $\|\vf-\vg\|\leq \epsilon$. Observe that $\| f \|_2 \leq \|g\|_2 + \|f-g\|_2$ and that by \autoref{lem:two-norms}
$$\|\Im(f)-\Im(g)\|_\infty \;,\; \|\Re(f)-\Re(g)\|_\infty \leq  \|\vf-\vg\| \cdot \sqrt{\Nhom} = \epsilon \cdot \sqrt{\Nhom} .$$
Thus 
$$\|f-g\|_\infty \leq \sqrt{2}\cdot  \epsilon \cdot \sqrt{\Nhom}.$$
Since $\Re(f)$ has a large coefficient and $\|\vf-\vg\|\leq \epsilon < 1/10$ it follows that some coefficient in $\R(g)$ is at least $1/4$. 
Let $\vv\in\cH$ be such that $|g(\vv)|\geq\eta\cdot\|g\|_2$. Then,
\begin{eqnarray}
|f(\vv)| \geq |g(\vv)| - |(f-g)(\vv)| \geq \eta \cdot \|g\|_2- \sqrt{2}\cdot\epsilon \cdot \sqrt{\Nhom}  &\geq & \eta \cdot \|g\|_2-  \frac{1}{8}\eta \cdot 2^{n/2} \cdot e^{-r} \label{eq:1} \\
&\geq & \frac{1}{2}\eta  \cdot \|g\|_2 \label{eq:2}\\
&\geq & \frac{1}{4}\eta  \cdot \|f\|_2 \label{eq:3}
\end{eqnarray}
where \autoref{eq:1} holds because of the assumption in the lemma, \autoref{eq:2} follows from \autoref{lem:norm-coeff} using the fact that  some coefficient in $g$ is at least $1/4$. Indeed,  \autoref{lem:norm-coeff} implies that $\|g\|_2 \geq \frac{1}{4} 2^{n/2}\cdot e^{-r}$.
Finally, \autoref{eq:3} holds since 
\begin{eqnarray*}
\|f\|_2 &\leq&  \|g\|_2 +  \|f-g\|_2 \\
& = &  \|g\|_2 + \|\Re(f-g)\|_2 + \|\Im(f-g)\|_2 \\
&\leq & \|g\|_2 +  \|\Re(f-g)\|_\infty + \|\Im(f-g)\|_\infty \\ 
&\leq & \|g\|_2 + 2\epsilon \cdot \sqrt{\Nhom} \leq 2\|g\|_2.
\end{eqnarray*}
\end{proof}

\subsection{Robust hitting sets for polynomial varieties}

In this section we prove that for any variety of $n$-variate degree $r$ polynomials of polynomial dimension and exponential degree there exists an $\exp(-\poly(nr))$ robust hitting set of polynomials size.

As before we think of every point $\vf\in \C^\Nhom$ as vector of coefficients of a homogeneous $n$-variate, degree $r$ polynomial $f(\vx)$. 

\begin{theorem}[Robust hitting sets for varieties]\label{thm:robust-general}
Let $V\subset \C^\Nhom$ be a variety of dimension $d$ and degree $D$ satisfying the assumption of \autoref{lem:robust-net}. 
Let $\eta = 2^{-n}\cdot \frac{1}{2\cdot \left(C_{CW}\cdot n\cdot r\right)^r} $ and $\delta = \frac{\eta}{(16nr^2)^{2n+1}}$. There exists an $(\eta/4)$-robust hitting set $\cH$ for $V$ satisfying  $\cH\subset G_\delta^\C$ of size $$|\cH| =\max\{2\log D , 12r(n+r)\cdot d\}.$$
\end{theorem}

\begin{proof}
Let $k=\max\{2\log D , 18(n+r)\cdot d\}$. Sample $k$ points $\vv_1,\ldots,\vv_k\in G_\delta^\C$ uniformly and independently at random. Set $\cH = \{\vv_1,\ldots,\vv_k\}$.

Let $\epsilon = \left(\frac{1}{\Nhom}\right)^{r}$ and $E\subset V(n,s,r)\cap [-(1-\frac{1}{N}),1-\frac{1}{N}]^N_\C$ be the $\epsilon$-net guaranteed by \autoref{cor:eps-net}. From \autoref{thm:CW-discrete-complex} (by taking $\alpha=\eta\cdot \|g\|_2$) it follows by the union bound that the probability that there exists $\vg\in E$ such that for every $i$, $$|g(\vv_i)|\leq \left(\eta - \frac{1}{2}\cdot \delta \cdot (16nr^2)^{2n+1}\right)\cdot \|g\|_2 = \eta/2 \cdot \|g\|_2,$$ is at most 
\begin{eqnarray*}
\left(C_{CW}\cdot r \cdot (2\eta)^{1/r}\right)^k \cdot |E| &\leq & \left(C_{CW}\cdot r \cdot (2\eta)^{1/r}\right)^k \cdot D\cdot (75\Nhom^2 /\epsilon^2)^{d+1} \\
&< & 2^{-nk/r}\cdot \left(\frac{1}{n} \right)^k \cdot D\cdot \Nhom^{12rd}\\
&< & 2^{-nk/r}\cdot \left(\frac{1}{n} \right)^k \cdot D\cdot 2^{12rd(n+r)}\\
&\leq& 1,
\end{eqnarray*}
where the last inequality follows from the definition of $k$.
Notice that our choice of $\eta$ and $\epsilon$ satisfy the condition in \autoref{lem:robust-net}, that is,  $$10\cdot \epsilon \cdot \sqrt{\Nhom} < \frac{1}{8}\eta \cdot 2^{n/2} \cdot e^{-r} < \frac{1}{4}.$$
The claim of the theorem now follows from \autoref{lem:robust-net}.
\end{proof}

\begin{corollary}[Robust hitting sets for $V(n,s,r)$]\label{cor:robust-for-easy}
There exists a constant $c$ such that for every integers $n,s,r$, for $\eta = 2^{-n}\cdot \frac{1}{2\cdot \left(C_{CW}\cdot nr\right)^r} $ and $\delta = \frac{\eta}{(16nr^2)^{2n+1}}$,  there is an  $\eta/4$-robust hitting set $\cH\subset G_\delta^\C$, for $V(n,s,r)$, of size $$|\cH| \leq (nsr)^c.$$ 
\end{corollary}

\begin{proof}
The proof follows immediately from applying \autoref{thm:robust-general} to $V(n,s,r)$ using the estimates given in \autoref{cor:variety-universal}.
\end{proof}

We note that the proof above gives a robust hitting set $\cH$ for $V(n,s,r)$ whose points come from $\C^n$. We obtain a hitting set for $V(n,s,r)$ over $\R$ by using a simple trick. 

\begin{theorem}[Robust-hitting sets for $V(n,s,r)$ over $\R$]\label{thm:robust-over-R}
Let $\cH$ be an $\epsilon$-robust hitting set for $V(n,s,r)$. For each $\vv = \va+\iota \cdot \vb \in \cH$ and an integer $k$ let $\vv_k = \va+k\cdot \vb$. Set $\cH_\R \triangleq \{\vv_k \mid \vv\in \cH \text{ and } k\in\{0,\ldots,r\}\}$. It holds that $\cH_\R$ is an $\frac{\epsilon}{(r+2)!}$-robust hitting set for $V(n,s,r)$. 
\end{theorem}

\begin{proof}
The fact that $\cH_\R$ hits  $V(n,s,r)$ is simple. Let $z$ be a new variable. For each $\vv = \va+\iota \cdot \vb \in \cH$ and $f\in V(n,s,r)$ consider a new univariate polynomial $F_\vv(z) = f(\va+z\cdot \vb)$. Clearly, there is some $\vv\in \cH$ for which $F_\vv(z)\not\equiv 0$ as setting $z=\iota$ gives $F_\vv(\iota)=f(\vv)$. 

Let $c_k$, for $k\in \{0,\ldots,r\}$, be constants satisfying that for every degree $r$ polynomial $g$ it holds that $g(\iota) = \sum_{k=0}^{r}c_k\cdot g(k)$. Such constants exist by simple interpolation. Furthermore, we have that $|c_k| \leq (r+1)!$. Indeed,  note that $c_k$ is the $k$'th entry in the result of the following matrix-vector product: The vector has length $(r+1)$ and its $k$'th entry is $\iota^k$. That is, denoting the vector with $\vw$ we have that $w_k = \iota^k$. The matrix is the inverse matrix of the Vandermonde matrix $A$ whose $(k,\ell)$-entry, for $k,\ell \in \{0,\ldots,r\}$, is $k^\ell$, where $0^0=1$. To get the bound on $|c_k|$ we observe that the $\ell$'th column of the inverse matrix is given by the coefficient vector of the polynomial $\frac{\prod_{k\neq \ell}(x-k)}{\prod_{k\neq \ell}(\ell-k)}$. The (very) crude bound we gave on $|c_k|$ follows easily.

To see that $\cH_\R$ is a robust hitting set we note that for $\vv\in \cH$ for which $|f(\vv)|\geq \epsilon \cdot \|f\|_2$ we have that 
\begin{eqnarray*}
\epsilon\cdot \|f\|_2 \leq \left|f(\vv)\right| = \left|F_\vv(\iota)\right| &=& \left|\sum_{k=0}^{r} c_k F_\vv(k)\right| \\& =&  \left|\sum_{k=0}^{r} c_k f(\vv_k)\right| \\& \leq& (r+1)\cdot \max_k |c_k| \cdot \max_k |f(\vv_k)| \\
&\leq &(r+2)! \cdot \max_k |f(\vv_k)| \;.
\end{eqnarray*}
\end{proof}

We now state the consequence for the robust hitting set constructed in \autoref{cor:robust-for-easy}. We denote $G_{\delta,r} \triangleq \{\va+k\cdot \vb \mid \va,\vb\in G_\delta \text{ and } 0\leq k\leq r\}$.

\begin{corollary}[Robust-hitting sets for $V(n,s,r)$ over $\R$]\label{cor:robust-over-R}
There exists a constant $c$ such that for every integers $n,s,r$, for $\eta = 2^{-n}\cdot \frac{1}{20\cdot \left(C_{CW}\cdot nr^2\right)^r} $ and $\delta = \frac{\eta}{(16nr^2)^{2n+1}}$,  there is an  $\eta/4$-robust hitting set $\cH\subset G_{\delta,r}$, for $V(n,s,r)$, of size $$|\cH| \leq (nsr)^c.$$ 
\end{corollary}

\begin{proof}
The proof follows immediately from combining \autoref{cor:robust-for-easy} with \autoref{thm:robust-over-R} and observing that $(r+2)!<10 r^r$. 
\end{proof}

\section{Existential theory of the reals}\label{sec:logic}

We will need the following theorem regarding the decidability of existential formulas over the reals. To keep this manuscript at a reasonable length we will not give a formal definition of sentences and formulas over the reals. The interested reader is referred to \cite{BasuPollackRoy}. Intuitively, formulas are constructed as follows. The atoms are polynomial equalities ``$f(\vx)=0$'' or inequalities ``$f(\vx)\geq 0$''. From them we build formulas in a similar fashion to the way we build formulas in first order logic using the connectives $\neg,\vee,\wedge$ and the quantifiers $\exists, \forall$ (however, in the existential theory we only allow existential quantifiers). For a set of polynomials $\cF\subset \R[\vx]$, an $\cF$-formula is a formula in which all the polynomials appearing in the atoms are from $\cF$.

\begin{theorem}[Existential theory of the reals in $\PSPACE$ \cite{Canny88}]
Let $\cF\subset \R[\vx]$ be a set of $\poly(n)$ polynomials each of degree at most $r=\poly(n)$ and let $\exists x_1 \exists x_2 \ldots \exists x_n F(x_1,\ldots,x_n)$ be a sentence where $F(\vx)$ is a quantifier free $\cF$-formula. There is a $\PSPACE$ algorithm for deciding the truth of the sentence, where the size of the input to the algorithm is the bit complexity of the formula $F$.
\end{theorem}

\subsection{Formulas capturing computations by algebraic circuits}

\begin{lemma}[Computation by the universal circuit]\label{lem:FO-universal-circuit}
Let $n,s,r$ be natural numbers and $\epsilon$ be a rational number with $\poly(n)$ bit complexity.
For real vectors $\vv,\va$ over the reals there exists an existential sentence over the reals, $\phi(\vecv,\veca,\epsilon)$,  such that 
$\phi(\veca,\vecv,\epsilon)$ is true iff the polynomial computed by the universal circuit for size $s$ and degree $r$, whose auxiliary variables are set to $\veca$, evaluates on  input $\vecv$, in absolute value, to at least $\epsilon$. That is, for $\Psi(\vx,\vy)$ as in \autoref{thm:universal}, $|\Psi(\vv,\va)|\geq \epsilon$.
\end{lemma}

\begin{proof}
Let $\Psi(\vx,\vy)$ be the universal circuit for size $s$ and degree $r$ (and $n$ variables), where $\vy$ are the auxiliary variables. 
For each gate $u$ of $\Psi$ let  $\Psi_u(\vx,\vy)$ be the polynomial computed at $u$. For each gate $u$ of $\Psi$ let $z_u$ be a new variable and denote with $z_o$ be the variable corresponding to the output gate. 

For each internal gate we assign a polynomial equation as follows: If $u$ is an addition gate with children $w_1$ and $w_2$ then we assign the equation $z_u - (z_{w_1}+z_{w_2})=0$ to $u$. If it is a multiplication gate with children $w_1,w_2$ then we assign the equation $z_u - z_{w_1}\cdot z_{w_2} =0$ to $u$. In addition we assign the inequality $z_o^2 \geq \epsilon^2$ to the output gate. For an input gate $u$ corresponding to a variable $x_i$ consider the equation $z_u - v_i=0$. For an input gate corresponding to $y_i$ we have the equation $z_u - a_i=0$.

Let $\cF$ be the set of all equalities and inequalities constructed above. Consider the sentence 
$$\phi(\vv,\va,\epsilon)\triangleq \exists \vz \bigwedge_{g\in \cF} g(\vz)\;,$$
where $\exists \vz$ is a short hand for writing $\exists z_u$ for all gates $u$ in $\Psi$.
It is not hard to see that the exists an assignment to the $z_u$ satisfying this sentence iff $|\Psi(\vv,\va)|\geq \epsilon$. 
\end{proof}

The next lemma shows that deciding whether a polynomial computed by the universal circuit evaluates to at least $1$ on some input from $[-1,1]^n$ can be done in $\PSPACE$.

\begin{lemma}\label{lem:deciding-nonzero}
Let $\Psi(\vx,\vy)$ be the universal circuit for size $s$ and degree $r$ (and $n$ variables).
Let $f(\vx)\in \R[\vx]$ be computed by $\Psi(\vx,\vy)$ when assigning  $\va$ to the auxiliary variables. That is, $f(\vx)=\Psi(\vx,\va)$. Given $n,s,r$ in unary encoding there is a $\PSPACE$ algorithm for deciding whether there exists $\vv\in[-1,1]^n$ on which $|f(\vv)|\geq 1$.
\end{lemma}

\begin{proof}
Let $\cF$ be the set of polynomials in the definition of $\phi(\vv,\va,1)$ in \autoref{lem:FO-universal-circuit}. Define
$$\psi(\va,1) \triangleq \exists \vv \exists \vz \left(\bigwedge_{g\in \cF} g(\vz) \right) \bigwedge \left( \bigwedge_i \left((1-v_i^2) \geq 0\right) \right) \;.$$
It is not hard to see that $\psi(\va,1)$ is true iff there exists $\vv \in [-1,1]^n$ such that $|f(\vv)|=|\Psi(\vv,\va)|\geq 1$.
\end{proof}

\section{Construction of a hitting set for $\VPbar$ in $\PSPACE$}\label{sec:proof}

\begin{algorithm}
  \caption{: Finding a robust hitting set}
  \label{alg:robust-find}
\begin{algorithmic}[1]
  \Require{Parameters $n,s,r$}
  \State{Let $\eta,\delta$ be as in \autoref{cor:robust-over-R}.}
  \State{Set $\epsilon = \frac{1}{4}\cdot \eta  \cdot \left( \frac{1}{32\cdot n\cdot r^2} \right)^n $.}
  \State{Let $m=(nsr)^c$ (as in \autoref{cor:robust-over-R}).}
  \State{Let  $\vv_1,\ldots,\vv_m \in G_{\delta,r}$ so that $(\vv_1,\ldots,\vv_m)$ is the lexicographically first string in $G_{\delta,r}^m$.}
  \While{Robust set not found yet}
  \State{Check whether there is $\va$ for which $\psi(\va,1)$ (of \autoref{lem:deciding-nonzero}) is true and for all $i\in [m]$, $\phi(\vv_i,\va,\epsilon)$ (of \autoref{lem:FO-universal-circuit}) is false.}
  \State{If no solution $\va$ is found then halt and return $\cH = \{\vv_1,\ldots,\vv_m\}$.}
    \State{Otherwise, move to the next $\vv_1,\ldots,\vv_m \in G_{\delta,r}$}
  \EndWhile
\end{algorithmic}
\end{algorithm}

\begin{theorem}[Main theorem]\label{thm:PIT-PSPACE}
Let $c$ be a constant as in  \autoref{cor:robust-over-R}).
\autoref{alg:robust-find} returns an $\epsilon= \frac{1}{4}\cdot \left( \frac{1}{C_{CW}\cdot n\cdot r}\right)^r \cdot \left( \frac{1}{32\cdot n \cdot r^2} \right)^n$ robust hitting set of size $(nsr)^c$ for $V(n,s,r)$ and can be executed in $\PSPACE$, given $n,s,r$ in unary encoding.
\end{theorem}

\begin{proof}

We first prove that the algorithm always returns some $\cH$ and then prove that any $\cH$ returned by the algorithm is an $\epsilon$ robust hitting set.

\paragraph{The algorithm always outputs some set:}
\autoref{cor:robust-over-R} guarantees that for $m=(nsr)^c$ there exist $\vv_1,\ldots,\vv_m \in G_{\delta,r}^m$ so that  $\cH = \{\vv_1,\ldots,\vv_m\}$ is an $\frac{1}{4}\cdot \frac{2^{-n}}{20\cdot \left( C_{CW}\cdot n\cdot r^2\right)^r} $ robust hitting set for $V(n,s,r)$.
Assume that our while loop reached that set $\vv_1,\ldots,\vv_m$ (that is, the algorithm did not output any set so far). 

Let $\va$ be any assignment for the auxiliary variables of the universal circuit and let $f = \Psi(\vx,\va)$. If  $\psi(\va,1)$ (of \autoref{lem:deciding-nonzero}) is true then for some $\vu\in [-1,1]^n$, $|f(\vu)|\geq 1$. As $\|f\|_\infty \geq 1$, \autoref{cor:infty-to-one} implies that 
$$\|f\|_2 \geq \frac{1}{2^{2n+2}}\cdot V(n,\frac{1}{4r^2}) \geq \left( \frac{1}{32\cdot n\cdot r^2} \right)^n \,.$$
As $\cH$ is an $\frac{\eta}{4} = \frac{1}{4}\cdot \frac{2^{-n}}{20\cdot \left( C_{CW}\cdot n\cdot r^2\right)^r} $ robust hitting set  for $V(n,s,r)$, for some $i$, 
$$|f(\vv_i)| \geq \frac{1}{4}\cdot \frac{2^{-n}}{20\cdot \left( C_{CW}\cdot n\cdot r^2\right)^r}  \cdot \|f\|_2 \geq \frac{1}{4}\cdot \frac{2^{-n}}{20\cdot \left( C_{CW}\cdot n\cdot r^2\right)^r}  \cdot \left( \frac{1}{32\cdot n\cdot r^2} \right)^n =\epsilon \;.$$
Thus, $\phi(\vv_i,\va,\epsilon)$ will return true. In particular, no solution $\va$ will be found and so the algorithm will return $\cH$ if it did not halt before reaching this particular $\cH$. Finally, note that there are polynomials $f$ for which $\psi(\va,1)$ is true. Indeed, if $f\not \equiv 0$ then there is some multiple of it that at some point in $[-1,1]^n$ will get value at least $1$. By  \autoref{thm:universal}  this multiple of $f$ is also computed by the universal circuit.

\paragraph{Every output is a robust hitting set:}\sloppy
Assume that the algorithm returned some set $\cH=\{\vu_1,\ldots,\vu_m\}$. Let $f$ be a nonzero polynomial computed by a homogeneous algebraic circuit of size $s$ and degree $r$ and assume further that $\|f(\vu)\|_\infty =1$.\footnote{We can restrict our attention to such $f$'s as for any constant $\alpha$ the universal circuit also computes $\alpha\cdot f$. See \autoref{thm:universal}.} In particular there is some assignment $\va$ to the auxiliary variables of the universal circuit $\Psi$ so that $\Psi(\vx,\va)=f(\vx)$. Clearly, for this $\va$, $\psi(\va,1)$ is true. As $\cH$ was returned it means that for some $i$, $|f(\vu_i)| \geq \epsilon$. As $\|f\|_2 \leq \|f\|_\infty=1$ we get that 
$$|f(\vu_i)| \geq \epsilon \geq \epsilon \cdot \|f\|_2 \;.$$
As both size of the equation scale the same way when we multiply $f$ by a field element, this equation holds regardless of $\|f\|_\infty$.
In other words, $\cH$ is an $\epsilon$-robust hitting set for all polynomials that can be computed by size $s$ and degree $r$ homogeneous circuits and hence it is  an $\epsilon$-robust hitting set for $V(n,s,r)$.

\paragraph{Complexity:}
The fact that the algorithm can be run in $\PSPACE$ follows from the fact that all the vectors that are considered (and also $\epsilon$) have polynomial bit length and from \autoref{lem:FO-universal-circuit} and \autoref{lem:deciding-nonzero}.
\end{proof}

\bibliographystyle{customurlbst/alphaurlpp}
\bibliography{bibliography}

\begin{thebibliography}{GMQ16}

\bibitem[Bl{\"{a}}13]{DBLP:journals/toc/Blaser13}
Markus Bl{\"{a}}ser.
\newblock \href {http://dx.doi.org/10.4086/toc.gs.2013.005} {Fast Matrix
  Multiplication}.
\newblock {\em Theory of Computing, Graduate Surveys}, 5:1--60, 2013.

\bibitem[BPR06]{BasuPollackRoy}
Saugata Basu, Richard Pollack, and Marie-Fran\c{c}oise Roy.
\newblock {\em Algorithms in Real Algebraic Geometry}.
\newblock Springer-Verlag, 2006.

\bibitem[B{\"{u}}r04]{Burgisser04}
Peter B{\"{u}}rgisser.
\newblock \href {http://dx.doi.org/10.1007/s10208-002-0059-5} {The Complexity
  of Factors of Multivariate Polynomials}.
\newblock {\em Foundations of Computational Mathematics}, 4(4):369--396, 2004.

\bibitem[Can88]{Canny88}
John~F. Canny.
\newblock \href {http://dx.doi.org/10.1145/62212.62257} {Some Algebraic and
  Geometric Computations in {PSPACE}}.
\newblock In Janos Simon, editor, {\em Proceedings of the 20th Annual {ACM}
  Symposium on Theory of Computing, May 2-4, 1988, Chicago, Illinois, {USA}},
  pages 460--467. {ACM}, 1988.

\bibitem[CLO06]{CLO}
David~A. Cox, John Little, and Donal O'Shea.
\newblock {\em Ideals, Varieties, and Algorithms: An Introduction to
  Computational Algebraic Geometry and Commutative Algebra}.
\newblock Springer, 2006.

\bibitem[CW01]{CarberyWright}
Anthony Carbery and James Wright.
\newblock Distributional and L\^{}q norm inequalities for polynomials over
  convex bodies in R\^{}n.
\newblock {\em Mathematical Research Letters}, 8(3):233--248, 2001.

\bibitem[GMQ16]{GrochowMQ16}
Joshua~A. Grochow, Ketan~D. Mulmuley, and Youming Qiao.
\newblock \href {http://dx.doi.org/10.4230/LIPIcs.ICALP.2016.34} {Boundaries of
  {VP} and {VNP}}.
\newblock In Ioannis Chatzigiannakis, Michael Mitzenmacher, Yuval Rabani, and
  Davide Sangiorgi, editors, {\em 43rd International Colloquium on Automata,
  Languages, and Programming, {ICALP} 2016, July 11-15, 2016, Rome, Italy},
  volume~55 of {\em LIPIcs}, pages 34:1--34:14. Schloss Dagstuhl -
  Leibniz-Zentrum fuer Informatik, 2016.

\bibitem[HS80a]{DBLP:conf/stoc/HeintzS80}
Joos Heintz and Claus{-}Peter Schnorr.
\newblock \href {http://dx.doi.org/10.1145/800141.804674} {Testing Polynomials
  which Are Easy to Compute (Extended Abstract)}.
\newblock In Raymond~E. Miller, Seymour Ginsburg, Walter~A. Burkhard, and
  Richard~J. Lipton, editors, {\em Proceedings of the 12th Annual {ACM}
  Symposium on Theory of Computing, April 28-30, 1980, Los Angeles, California,
  {USA}}, pages 262--272. {ACM}, 1980.

\bibitem[HS80b]{HeintzSieveking}
Joos Heintz and Malte Sieveking.
\newblock Lower bounds for polynomials with algebraic coefficients.
\newblock {\em Theoretical Computer Science}, 11(3):321--330, 1980.

\bibitem[Koi96]{DBLP:journals/jc/Koiran96}
Pascal Koiran.
\newblock \href {http://dx.doi.org/10.1006/jcom.1996.0019} {Hilbert's
  Nullstellensatz Is in the Polynomial Hierarchy}.
\newblock {\em J. Complexity}, 12(4):273--286, 1996.

\bibitem[LL89]{LehmkuhlL89}
Thomas Lehmkuhl and Thomas Lickteig.
\newblock \href {http://dx.doi.org/10.1016/0304-3975(89)90141-2} {On the Order
  of Approximation in Approximative Triadic Decompositions of Tensors}.
\newblock {\em Theor. Comput. Sci.}, 66(1):1--14, 1989.

\bibitem[Mul17]{Mulmuley-GCT-V}
Ketan~D. Mulmuley.
\newblock {Geometric complexity theory V: Efficient algorithms for Noether
  normalization}.
\newblock {\em J. Amer. Math. Soc.}, 30(1):225–309, 2017.

\bibitem[Raz10]{DBLP:journals/toc/Raz10}
Ran Raz.
\newblock \href {http://dx.doi.org/10.4086/toc.2010.v006a007} {Elusive
  Functions and Lower Bounds for Arithmetic Circuits}.
\newblock {\em Theory of Computing}, 6(1):135--177, 2010.

\bibitem[San91]{Sansone}
Giovanni Sansone.
\newblock \href
  {http://gen.lib.rus.ec/book/index.php?md5=622ad96172b8583016f546f4ba5a0d8e}
  {{\em Orthogonal Functions}}.
\newblock Dover Books on Advanced Mathematics. Dover Publications, revised
  edition, 1991.

\bibitem[Sha88]{Shafarevich-II}
Igor~R. Shafarevich.
\newblock {\em Basic Algebraic Geometry 2}.
\newblock Springer-Verlag, 1988.

\bibitem[SY10]{SY10}
Amir Shpilka and Amir Yehudayoff.
\newblock Arithmetic Circuits: A survey of recent results and open questions.
\newblock {\em Foundations and Trends in Theoretical Computer Science},
  5(3-4):207--388, 2010.

\bibitem[Wil74]{Wilhelmsen}
Don~R. Wilhelmsen.
\newblock A Markov inequality in several dimensions.
\newblock {\em Journal of Approximation Theory}, 11(3):216--220, 1974.

\end{thebibliography}

\end{document}